\PassOptionsToPackage{bookmarks={false}}{hyperref}

\documentclass[sigconf]{acmart}

\usepackage{hyperref}

\usepackage{enumitem}

\usepackage{cmap}
\usepackage{listings}

\usepackage{graphicx}
\usepackage{proof}
\usepackage{url}
\usepackage{centernot}
\usepackage{courier}
\usepackage[ruled,linesnumbered,algo2e]{algorithm2e}  
\usepackage[noend]{algpseudocode}
\let\oldReturn\Return
\renewcommand{\Return}{\State\oldReturn}
\usepackage{float}
\usepackage{mathtools}
\usepackage[framemethod=TikZ]{mdframed}
\usepackage[]{inputenc}
\usepackage[T1]{fontenc}

\usepackage{placeins}
\usepackage{xspace}
\usepackage[font=footnotesize,labelfont=bf]{caption}
\usepackage[framemethod=TikZ]{mdframed}
\usepackage{booktabs}
\usepackage{fancyvrb}
\usepackage{relsize}
\graphicspath{{images/}}

\usepackage{pifont}
\newcommand{\cmark}{\ding{52}}
\newcommand{\xmark}{\ding{56}}
\usepackage{float}
\usepackage{subfig}
\usepackage{framed}
\usepackage{multirow}

\newcommand{\aeval}{\textsc{AE-VAL}\xspace}
\newcommand{\jkind}{\textsc{JKind}\xspace}

\newcommand{\jsyn}{\textsc{JSyn}\xspace}

\newcommand{\jsynvg}{\textsc{JSyn-vg}\xspace}

\newcommand{\viable}{{\mathsf {Viable}}}

\newcommand{\tuple}[1]{\langle #1 \rangle}

\newcommand{\vx}{\vec{x}}
\newcommand{\vy}{\vec{y}}

\newcommand{\realizable}{\textsc{realizable}\xspace}
\newcommand{\unrealizable}{\textsc{unrealizable}\xspace}
\newcommand{\skolems}{\textit{Skolem}}

\newcommand{\true}{\top}
\newcommand{\false}{\bot}
\newcommand{\Land}{\bigwedge}
\newcommand{\Lor}{\bigvee}
\newcommand{\valid}{\mathit{valid}}
\newcommand{\invalid}{\mathit{invalid}}
\newcommand{\subs}{\textit{validRegion}}

\newcommand{\geqp}{%
  \mathrel{\raisebox{-0.5ex}{$\scriptscriptstyle($}}%
  \geq
  \mathrel{\raisebox{-0.5ex}{$\scriptscriptstyle)$}}%
}

\newcommand{\leqp}{%
  \mathrel{\raisebox{-0.5ex}{$\scriptscriptstyle($}}%
  \leq
  \mathrel{\raisebox{-0.5ex}{$\scriptscriptstyle)$}}%
}
\newcommand{\such}{\,.\,}

\newcommand\eqdef{\mathrel{\stackrel{\makebox[0pt]{\mbox{\normalfont\tiny def}}}{=}}}

\newcounter{template}

\copyrightyear{2020}
\acmYear{2020}
\setcopyright{acmcopyright}\acmConference[ASE '20]{35th IEEE/ACM International Conference on Automated Software Engineering}{September 21--25, 2020}{Virtual Event, Australia}
\acmBooktitle{35th IEEE/ACM International Conference on Automated Software Engineering (ASE '20), September 21--25, 2020, Virtual Event, Australia}
\acmPrice{15.00}
\acmDOI{10.1145/3324884.3416586}
\acmISBN{978-1-4503-6768-4/20/09}

\begin{document}
\title{Synthesis of Infinite-State Systems with Random Behavior}

\author{Andreas Katis}
\orcid{0000-0001-7013-1100}
\email{katis001@umn.edu}
\affiliation{University of Minnesota, USA}
\author{Grigory Fedyukovich}
\orcid{0000-0003-1727-4043}
\email{grigory@cs.fsu.edu}
\affiliation{Florida State University, USA}
\author{Jeffrey Chen}
\email{chen4233@umn.edu}
\affiliation{University of Minnesota, USA}
\author{David Greve}
\email{david.greve@collins.com}
\affiliation{Collins Aerospace, USA}
\author{Sanjai Rayadurgam}
\email{rsanjai@umn.edu}
\affiliation{University of Minnesota, USA}
\author{Michael W. Whalen}
\email{whalen@cs.umn.edu}
\affiliation{University of Minnesota, USA}

\begin{abstract}
Diversity in the exhibited behavior of a given system is a desirable characteristic in a variety of application contexts. 
Synthesis of conformant implementations often proceeds by discovering witnessing Skolem functions, which are traditionally deterministic. In this paper, we present a novel Skolem extraction algorithm to enable synthesis of witnesses with random behavior and demonstrate its applicability in the context of reactive systems. The synthesized solutions are guaranteed by design to meet the given specification, while exhibiting a high degree of diversity in their responses to external stimuli. Case studies demonstrate how our proposed framework unveils a novel application of synthesis in model-based fuzz testing to generate fuzzers of competitive performance to general-purpose alternatives, as well as the practical utility of synthesized controllers in robot motion planning problems.
\end{abstract}

\maketitle

\section{Introduction}
\label{sec:intro}

\emph{Program synthesis} aims at automated generation of implementations that meet formal specifications. 
It has been thoroughly explored in various contexts, such as controller synthesis and automated program repair~\cite{alur2013syntax,neider2019learning,aeval-pbe,abate2017automated, nguyen2017connecting, chakraborty2018functional, gulwani2010dimensions}. 
The implementations are generated from of the specification's realizability and have the form of \emph{deterministic witnesses}.
Thus, by design they always compute (1) an output that meets the specification, and (2)
the same output for each particular input.
Determinism, however, prevents us from synthesizing systems that take advantage of \emph{randomness} to diversify their behavior%
\footnote{For the sake of brevity, throughout the paper, we refer to such systems using the adjective \emph{random} (e.g. \emph{random} system/design/witness/controller).}.
%
%
%
%
Advantages offered by these systems can be better understood when put into context of robot motion planning and fuzz testing.

\textbf{Fuzz testing}. Synthesis of random designs allows one to specify and create system-specific \textit{fuzzers}~\cite{manes2019art}. The idea is to follow a mindset similar to how \emph{model-based testing} techniques utilize the system-under-test (SUT) specification to generate test cases~\cite{utting2012taxonomy}. We propose to use the fragment of the model related to the SUT inputs to synthesize a fuzzer that repeatedly generates random (and sometimes malformed) tests. This fragment can be alternatively viewed as the fuzzer's specification, which can be further enriched with properties that dictate its behavior when certain testing objectives are met. For example, when a vulnerability is detected, we can limit the fuzzer's next generated test cases within a desired range around the test that exposed the issue. System coverage is also one such objective, where we can dictate how the fuzzer diversifies the generated tests through its specification, improving the chances of reaching previously unexplored system states. From a qualitative standpoint, synthesis in model-based fuzz testing can be considered as a viable high-level solution that does not require the user to create extensive corpora of tests. Furthermore, the synthesized fuzzers can be a strong, SUT-specific alternative to general-purpose model-based fuzzers.~\cite{holler2012fuzzing, aschermann2019nautilus, veggalam2016ifuzzer}.  

\textbf{Robot motion planning}. 
In coverage path planning problems, the goal is to maximize the area that a robot can cover while avoiding obstacles~\cite{galceran2013survey}. Furthermore, randomness can serve as an additional security barrier in avoidance games that involve adversaries with learning capabilities. A random strategy is inherently harder to infer and exploit. In the special case of infinite-state problems, it is an even bigger challenge, as the current state-of-the-art in automata learning is limited to finite-state problems~\cite{howar2018active}.


We treat systems in the aforementioned applications as so-called \emph{reactive systems}, which have to exhibit specification-compliant behavior against an unpredictable environment. Examples are commonly found in aviation, autonomous vehicles, and medical devices. 
Synthesis of random reactive designs is offered by
the recent
Reactive Control Improvisation (RCI)~\cite{fremont2017control, fremont2018reactivecontrolimprov} framework, but limited only to finite-state systems (i.e., over the boolean domain), relies on probabilistic analysis to determine the realizability of the specification, and its synthesized witnesses require further refinement to be applicable in real world scenarios.

We present a novel approach to synthesis of random infinite-state systems, whose corresponding specifications may involve constraints over the Linear Integer or Real Arithmetic theories (LIRA)~\cite{barrett2010smt} and thus not limited to finite-state systems. The intuition behind this effort is to allow reasoning, and consequently synthesis, over ranges of safe reactions instead of computing witnesses with deterministic responses. 

The pursuit of generality poses new challenges, for which we propose a novel Skolemization procedure to simulate randomness. 
We build on top of state-of-the-art reactive synthesis approach for deterministic systems called \jsynvg~\cite{katis2018validity}.
It iteratively generates a greatest fixpoint over system states that ensures the realizability of the given specification but offers only a brute and inflexible  strategy for witness extraction (via predetermined Skolemization rules)~\cite{simabs,aeval,aeval-pbe}.
Our key novelty is in a new algorithm that enables
replacing deterministic assignments in the Skolem functions with applications of \emph{uninterpreted random number generators}. Uninterpreted functions allow us to reason about solutions with random, broad, and most importantly, compliant behavior.

The new Skolem extraction algorithm preserves \jsynvg's important properties. 
Thus, the procedure remains completely automated, unlike previous work on infinite-state synthesis that requires additional templates, or the user's intervention~\cite{alur2013syntax, ryzhyk2009automatic, DBLP:conf/popl/BeyeneCPR14,finkbeiner2016synthesis}. More importantly, our work imposes no performance overheads over \jsynvg, remaining thus competitive with other state-of-the-art tools which could be considered for random synthesis~\cite{neider2019learning}. 
We implemented the Skolem extraction algorithm and applied it in two distinct case studies.

\smallskip
\noindent\textbf{Model-based fuzz testing}. We are the first to explore the applicability of reactive synthesis in fuzz testing. On a chosen set of applications designed for the DARPA Cyber Grand Challenge~\cite{DARPACGC, lee2015darpa}, the synthesized fuzzers performed competitively against well-established tools (\textsc{AFL}~\cite{zalewskiAFL}, \textsc{AFLFast}~\cite{bohme2017coverage}), both in terms of code coverage as well as exposing vulnerabilities.

\smallskip
\noindent\textbf{Robot motion planning}. We synthesized safe robot controllers that participate in avoidance games on both bounded and infinite arenas. Using simulation, we show how the synthesized controller leads to the robot being capable of avoiding its adversary while moving in random patterns. We demonstrate how the synthesized strategies are safe by design, no matter what bias is introduced at the implementation level. Furthermore, we showcase why randomness in the controller behavior is a mandatory feature, if synthesis is to be considered for coverage path planning problems.

\smallskip
To summarize, the contributions of this work are:

\begin{itemize}[leftmargin=*]
  \item the first complete formal framework that enables specification and synthesis of random infinite-state reactive systems;
  \item a novel Skolemization procedure that enables random synthesis with no performance overhead, by taking advantage of uninterpreted functions to reason about ranges of valid reactions;
  \item a novel application of synthesis in model-based fuzz testing, where we generated reactive fuzzers, yielding competitive results in terms of system coverage and vulnerability detection; and
  \item the application of synthesized random controllers in safety problems for robot motion planning, outlining important advantages over deterministic solutions.
\end{itemize} 

The rest of the paper is structured as follows.
Sect.~\ref{sec:background} provides the necessary formal background on which our work depends. 
Sect.~\ref{sec:ex} illustrates and Sect.~\ref{sec:aeval} describes in detail the algorithm for synthesis of random Skolem functions. 
The implementation is outlined in Sect.~\ref{sec:results} and
 the case studies are presented in Sect~\ref{sec:fuz} and Sect.~\ref{sec:robot}.
Finally, we discuss related work in Sect.~\ref{sec:related} and conclude in Sect.~\ref{sec:conclusion}.

\section{Background and Notation}
\label{sec:background}


A first-order formula $\varphi$ is satisfiable if there exists an assignment $m$, called a model, 
under which $\varphi$ evaluates to $\true$ (denoted $m\models\varphi$). 
If every model of $\varphi$ is also a model of $\psi$, then we write $\varphi \Rightarrow \psi$.
A formula $\varphi$ is called \emph{valid} if $\true \Rightarrow \varphi$.
For existentially-quantified formulas of the form $\exists y \such \psi(x, y)$, validity requires that each assignment of variables in $x$ 
can be \emph{extended} to a model of $\psi(x, y)$.
For a valid formula $\exists y \such \psi(x, y)$, a term $\mathit{sk}_{y}(x)$ is called a~\emph{Skolem}, if $\psi (x, \mathit{sk}_{y} (x))$ is valid.
More generally, for a valid formula $\exists \vy \such \psi(x, \vy)$ over a vector of existentially quantified variables $\vy$, there exists a vector of individual Skolem terms, one for each variable $\vy[j]$, where $0 < j \le N$ and $N = |\vy|$, such that:
%
$\true \Rightarrow \psi (x, \mathit{sk}_{\vy[1]} (x),\ldots, \mathit{sk}_{\vy[N]} (x))$. 

\subsection{Synthesis with \jsynvg}
We build on top of 
\jsynvg, a reactive synthesis procedure that takes formal specifications in the form of \emph{Assume-Guarantee contracts}.
%
%
Systems are described in terms of inputs $\vx$ and outputs $\vy$, using the predicate $I(\vy)$ to denote the set of initial outputs and $T(\vy, \vx, \vy')$ for the system's transition relation, where the next (primed) outputs $\vy'$ depend on the current input and state. 
\emph{Assumptions} $A(\vx, \vy)$ correspond to assertions over the system's current state, while the set of \emph{guarantees} is decomposed into constraints over the initial outputs $G_{I}(\vy)$, and guarantees $G_{T}(\vy, \vx, \vy')$ that have to hold over any valid transition (i.e., with respect to $T(\vy, \vx, \vy')$).

\begin{algorithm2e}[!t]
\small
\SetAlgoSkip{}
\KwIn{$A(\vx,\vy)$: {assumptions}, $G_{I}(\vy), G_{T}(\vy,\vx,\vy')$: {guarantees}}
\KwOut{$\tuple{\realizable, \skolems} / \unrealizable$}
\BlankLine
$F(\vy) \gets \true$\label{alg:init}\;
\While{$\true$}{
  $\phi \gets \forall \vx,\vy. \ (F(\vy) \land A(\vx, \vy) \Rightarrow \exists \vy'.G_{T}(\vy, \vx, \vy') \land F(\vy'))$\label{alg:ae1}\;
  $\tuple{\mathit{valid}, \subs(\vx, \vy), \skolems} \gets \aeval(\phi) \label{alg:val1}$\;
  \uIf{$\mathit{valid}$\label{alg:val2}}{
            \lIf{$\exists \vy . G_{I}(\vy)\! \land\! F(\vy)$}{%
				\textbf{return} $\tuple{\realizable, \skolems}$\label{alg:issat}}
			\lElse{%
		 		\textbf{return} $\unrealizable$\label{alg:unreal}}
  }
  \lElse{%
	$F(\vy) \gets F(\vy) \land \lnot \textsc{ExtractUnsafe}(\subs(\vx, \vy))$\label{alg:rem}}
}
\caption{\jsynvg $\Big(A(\vx, \vy), G_{I}(\vy), G_{T}(\vy, \vx, \vy')\Big)$, cf.~\cite{katis2018validity}.}
\label{alg:synthesis}
\end{algorithm2e}

The algorithm behind \jsynvg performs a realizability analysis to determine the existence of a greatest fixpoint of states meeting the contract, that can lead to an implementation. Furthermore, the computed fixpoint can be directly used for the purposes of synthesis, as it precisely captures a collection of system output constraints which, when instantiated, define safe reactions. Formally, the computed fixpoint is a set of \emph{viable} outputs, guaranteed to preserve safety by requiring that a valid transition to another viable output is always available.
\begin{align}
\begin{split}
  \viable(\vy) &\eqdef
  \forall \vx, \vy. (A(\vx, \vy) \Rightarrow \exists \vy'.~ G_T(\vy, \vx, \vy')
\land \viable(\vy'))
\label{eq:viable}
\end{split}
\end{align}

The coinductive definition of viable states is sufficient to prove the realizability of a contract, as long as the corresponding decision procedure can find a viable output that satisfies the initial guarantees $G_I(\vy)$:
\begin{equation*}
\exists \vy. G_I(\vy) \land \viable(\vy)
\label{eq:nonempty}
\end{equation*}

Given a proof of the contract's realizability, the problem of synthesis is formally defined as the process of computing an initial output \ $\vy_{\emph{init}}$ and a function $f(\vx, \vy)$ such that $G_{I}(\vy_{\emph{init}})$ and $\forall \vx, \vy. \viable(\vy) \Rightarrow \viable(f(\vx, \vy))$ hold true. 

Alg.~\ref{alg:synthesis} summarizes \jsynvg. It begins with the generic candidate fixpoint $F(\vy) = \true$ and solves
 the $\forall\exists$-formula $\phi$ for the validity (line~\ref{alg:val1}) that corresponds to the definition of viable outputs in Eq.~\ref{eq:viable}. If $\phi$ is valid and an output vector in $F(\vy)$ exists that satisfies the initial guarantees, then the contract is declared realizable, and a witnessing Skolem term is extracted.
If $\phi$ is invalid, the algorithm extracts the largest subset of $F(\vy) \land A(\vx, \vy)$, denoted $\subs(\vx, \vy)$, such that the following formula is valid:
$$
\forall \vx,\vy. \ (\subs(\vx, \vy) \Rightarrow \exists \vy'.G_{T}(\vy, \vx, \vy') \land F(\vy'))
$$
Due to the possibility of $\subs(\vx, \vy)$ strengthening the assumptions $A(\vx, \vy)$, we additionally extract a set of constraints over unsafe states (\textsc{ExtractUnsafe}) from $\subs(\vx, \vy)$. The negation of this set is then added as a new conjunct to the candidate $F(\vy)$ and the algorithm iterates until either $\phi$ is valid or $F(\vy) = \false$. For further details, we refer the reader to the original paper on \jsynvg~\cite{katis2018validity}.

\subsection{Realizability and Synthesis with \aeval}
\label{sec:decomp}

The greatest fixpoint algorithm described by \jsynvg uses \aeval, an algorithm to determine the validity of $\forall\exists$-formulas and to generate (deterministic) witnesses in the form of Skolem terms. The latter feature is also the point of interest behind this work, as randomness can be introduced through replacing the part of \aeval's Skolemization strategy with our new proposed algorithm.

Alg.~\ref{alg:ae_val} gives a brief pseudocode of \aeval%
\footnote{Note that for simplicity of presentation, in the pseudocode we assume a single existentially quantified variable $y$ (however, the algorithm  and the implementation can handle any vector $\vy$).}.
The idea is to enumerate all models of $\vx$ and to extend each of them to a model of $y$.
Because a naive enumeration would be endless, \aeval generates a sequence of Model-Based Projections (MBPs)~\cite{DBLP:conf/lpar/BjornerJ15} each of which groups models of $\vx$.
Formally, an MBP for model $m$ is a formula $P(\vx)$, such that $m\models P(\vx)$, and $P\Rightarrow\exists y. \psi(\vx, y)$.
To create it, \aeval gathers all literals of $\psi$ which are evaluated to true by $m$ (line~\ref{alg:lit}).
These literals are further referred to as \emph{Skolem constraints} $\pi$.
In linear arithmetic, each Skolem constraint is composed only of arithmetic relations, linear combinations over $\vx$ and $y$, and numeric constants.
Finally, to obtain an MBP $\mathit{pre}$, \aeval just eliminates $y$ from the conjunction of Skolem constraints (line~\ref{alg:proj}).

The \textsc{ExtractSk} procedure, used for Skolem extraction, implements an inflexible strategy to transform Skolem constraints to local Skolem terms (we refer the reader to the original paper on \aeval for further details~\cite{aeval}).
The final Skolem term has a form of a \emph{decision tree},  where preconditions are placed on the nodes and local Skolem terms (i.e., outputs of \textsc{ExtractSk}) are on the leaves, i.e., the nested if-then-else structure ($\mathit{ite}(\cdot)$):
\begin{align}\notag
sk_{y}(\vec x) \eqdef \mathit{ite}(\mathit{pre}[1], sk_{1, y}(\vec x), \mathit{ite}(\mathit{pre}[2], sk_{2, y}(\vec x), \ldots , \cr
 (\mathit{ite} (\mathit{pre}[M - 1], sk_{M-1, y}(\vec x), sk_{M, y}(\vec x)))))
\end{align}

Finally, in the case that the input formula $\forall \vx \exists y \such \psi(\vx,y)$ is invalid, \aeval returns $\Lor\limits_{i = 1}^{M - 1} {\mathit{pre}[i](\vx)}$ as the formula's maximal \emph{region of validity}, i.e., the maximal set of models of the universally quantified variables for which the formula becomes valid. This region is used by \jsynvg in order to further refine the candidate fixpoint during each of its iterations (Alg.~\ref{alg:synthesis}, line~\ref{alg:rem}).

\begin{algorithm2e}[t!]
\small
\SetAlgoSkip{}
\SetKwFor{While}{while}{do}{}
\SetKw{KwContinue}{continue}
\KwIn{$\forall \vx \exists y \such \psi(\vx,y)$}
\KwData{MBPs $\mathit{pre}$, Skolem constraints $\pi$, Skolem terms $\mathit{sk}$}
\KwOut{Return value $\in \{\valid, \invalid\}$ of ${\forall \vx \exists y \such \psi(\vx,y)}$, $\subs$, $\skolems $}
\BlankLine
$M \gets 1$\;
\While{$\true$}{
\uIf(\label{alg:returnUnsat}){$\Land\limits_{i = 1}^{M - 1} \neg{\mathit{pre}[i] (\vx)} \Rightarrow \false$} 
    {\textbf{return} $\tuple{\valid, \true, \textsc{DecisionTree}(\mathit{pre}, \mathit{sk})}$;}%
\uIf(\label{alg:returnSat}){$\exists m \models \psi(\vx,y) \land \Land\limits_{i = 1}^{M - 1} \neg{\mathit{pre}[i](\vx)}$}
{
  $\pi[M](\vx,y) \gets \{\ell \mid \ell \in \textsc{literals}(\psi) \land m\models\ell\}$\label{alg:lit}\;
  $\mathit{pre}[M](\vx) \gets \textsc{QElim}(y, \pi[M](\vx,y))$\label{alg:proj}\;
  $\mathit{sk}[M](\vx, y) \gets \textsc{ExtractSk}(\vx,y, \Land\limits_{c \in \pi[M](\vx, y)} c)$\;
  $M \gets M + 1$\;
}
\lElse    {\textbf{return} $\tuple{\invalid, \Lor\limits_{i = 1}^{M - 1} {\mathit{pre}[i](\vx)}, \varnothing}$}
}
\caption{\aeval \Big($\forall \vx \exists y \such \psi (\vx,y)$\Big), cf.~\cite{simabs,aeval}.}
\label{alg:ae_val}
\end{algorithm2e}

\section{Random Synthesis - Motivating Example}
\label{sec:ex}

In this section, we demonstrate a complete run of the synthesis procedure and show how the standard synthesized witnesses are unable to exhibit random behavior. As an example, we use a safety robot motion planning problem from Neider et al.~\cite{neider2019learning}. In this problem, a robot is placed on a one-dimensional grid with two players, the environment and the system, controlling its movement. Each player can choose to either move the robot left or right, or not move it at all (we refer to these choices using the values ${-1,0,1}$). The robot starts at $\mathit{position} = 0$, and the safety property for the system is to retain the robot in the area of the grid for which $\mathit{position} \geq 0$.

Fig.~\ref{fig:contract} shows an Assume-Guarantee contract for the example, described in the Lustre language. The contract has the singleton input $\vec x = \{x\}$ (internally identified by the \texttt{-\,-\%REALIZABLE} statement) and the outputs $\vec y = \{y, \mathit{position}\}$\footnote{Variable $\mathit{property}$ is local to the contract. Formally, local variables are treated as system outputs.}. The contract assumption is that the environment will only make legal choices, i.e., $A(\vx, \vy) \eqdef -1 \leq x \land x \leq 1$. The initial guarantee refers to the initial position of the robot and the system choices for movement, i.e., $G_{I}(\vy) \eqdef (\mathit{position} = 0) \land (-1 \leq y \land y \leq 1)$. On the other hand, the transitional guarantee captures the safety property along with the stateful computation step for the new position, i.e., $G_{T}(\vy, \vx, \vy') \eqdef (\mathit{position}' = \mathit{position} + x + y') \land (\mathit{position}' \geq 0) \land (-1 \leq y' \land y' \leq 1)$ (the transition relation for $\mathit{position}$ is defined using Lustre's  \texttt{->} and \texttt{pre} operators).\footnote{Intuitively, the problem is formalized in a manner where the position of the robot is updated after both players make a choice, with the system reacting to the choice of the environment.} The safety properties are captured by \texttt{ok1} and \texttt{ok2} (declared as such using \texttt{-\,-\%PROPERTY}).

\begin{figure}[!t]
\centering
\begin{lstlisting}[basicstyle=\small]
node onedim(x,y : int) returns();
var
  ok1, ok2 : bool;
  position : int;
let
  assert x >= -1 and x <= 1;
  position =  0 -> (pre(position) + x + y);
  ok1 = y >= -1 and y <= 1;
  ok2 = position >= 0;

  --%PROPERTY ok1;
  --%PROPERTY ok2;
  --%REALIZABLE x;
tel;
\end{lstlisting}
\caption{Assume-Guarantee contract in Lustre.}
\label{fig:contract}

\begin{lstlisting}[basicstyle=\small]
void skolem() {
  if(position + x == 1) {
    y = -1;
  } else if (position + x >= -1 &&
  	position + x <= 0) {
    y = - (position + x);
  } else {
    y = 0;
  }
}
\end{lstlisting}
\caption{Synthesized deterministic witness in C.}
\label{fig:detskolem}

\end{figure}

The procedure begins with a call to \jsynvg using the contract as its input. The contract is realizable and the greatest fixpoint of viable states is $F(\vy) \eqdef \mathit{position} \geq 0$. \aeval declares that the formula $\phi \eqdef \forall \vec y, \vec x.(A(\vx,\vy) \land F(\vy) \Rightarrow \exists \vy'. G_{T}(\vy,\vx,\vy') \land F(\vy')$ is valid and extracts a Skolem term as a witness. Fig.~\ref{fig:detskolem} presents a direct translation of the function to C. The synthesized implementation behaves in a deterministic way under the following conditions: 

\begin{enumerate}
\item whenever $\mathit{position} + x = 1$, the system chooses to move left ($y = -1$);
\item if $\mathit{position} + x$ equals 0 or -1, then the system chooses to do nothing or move right, respectively ($y = - (\mathit{position} + x)$);
\item for any other case, the system chooses to do nothing ($y = 0$).
\end{enumerate}

While the implementation preserves safety, the set of possible actions are limited due to the deterministic assignments to the output $y$. Interestingly, for this particular implementation the system forces the robot to go back to positions that are dangerously close to the unsafe region! Similarly, the corresponding solution by Neider et al. is the winning set $[0,3)$, which would translate to implementations where the system would never move the robot beyond $position = 2$. Nevertheless, implementations exist for which the system can exercise a broader set of behaviors. For this example in particular, when either condition (1) or (3) is true, the system can freely choose any possible move action without violating the safety properties. Fig.~\ref{fig:skolem} shows an implementation that can (theoretically) exercise any such possible assignment (we explain why in Sect.~\ref{sec:alg}). In the following sections, we present a new method to synthesize a random witness that can, in theory, provide all such possible permutations using a single implementation.

\begin{algorithm2e}[b!]
\small
\SetAlgoSkip{}
\SetInd{0.4em}{0.4em}
\KwIn{Variables $\vx, y$, Skolem constraints $\pi(\vx,y) =  \Land\limits_{r \in E \cup D \cup G \cup GE \cup L \cup LE}r(\vx, y) $}
\KwOut{Term $\mathit{sk}$, such that $(y = \mathit{sk}(\vx)) \Rightarrow \pi(\vx,y)$}
{$\ell_{\mathit{closed}} \gets \false, u_{\mathit{closed}} \gets \false$;\label{alg3:init}} \\
\lIf(\label{alg3:trivcase}){$E \neq \varnothing$}{\textbf{return} $\textsc{ASN}(e)$, s.t. $e\in{E}$}
\uIf(\label{alg3:boundchecks}){$G \cup GE \neq \varnothing$}{
  $\ell \gets \textsc{MAX}(G \cup GE)$; \\
  $\ell_{\mathit{closed}} \gets G = \varnothing \lor  \textsc{MAX}(G) < \textsc{MAX}(GE)$;}
\uIf(){$L \cup LE \neq \varnothing$}{
  $u \gets \textsc{MIN}(L \cup LE)$; \\
  $u_{\mathit{closed}} \gets L = \varnothing \lor  \textsc{MIN}(L) > \textsc{MIN}(LE)$;\label{alg3:boundchkend}}
\lIf(\label{alg3:lowequp}){$\ell(\vx) = u(\vx)$}{\textbf{return} $\ell$}
  $H \gets \{\textsc{ASN}(d) \mid d \in D\}$\label{alg3:extractRHS}\; 
  \lIf(\label{alg3:diseqbothundef}){$\ell = \emph{undef} \land u = \emph{undef}$}{\textbf{return} $f_{\textsc{rng}}(H, \true, \true,-\infty,+\infty)$}
  \lIf(){$\ell = \emph{undef}$}{\textbf{return} $f_{\textsc{rng}}(H, \true, u_{\mathit{closed}},-\infty,u)$\label{alg3:defu}}
  \lIf(){$u = \emph{undef}$}{\textbf{return} $f_{\textsc{rng}}(H, \ell_{\mathit{closed}}, \true,\ell,+\infty)$\label{alg3:deflt}}  
\textbf{return} $f_{\textsc{rng}}(H, \ell_{\mathit{closed}}, u_{\mathit{closed}}, \ell, u)$\label{alg3:finalret}
\caption{$\textsc{ExtractSk}(\vx, y, \pi)$}
\label{alg:extractsk}
\end{algorithm2e}

\section{Synthesizing random designs} 
\label{sec:aeval}

The standard Skolem term extraction algorithm in \aeval does not support the generation of Skolem functions with random variable assignments. In this section, we present a new procedure to compute witnesses that can be used to simulate random behavior. 

\subsection{Overview}

Our proposed algorithm preserves the overall structure of \aeval as well as the soundness of its results~\cite{aeval}. The main idea is to replace the deterministic assignments that eventually appear in the leaves of the generated decision tree with applications of uninterpreted functions, which when translated at the implementation level, can be viewed as function calls to a user-defined random number generator. We refer to these functions as \textit{uninterpreted random number generators}:

\begin{definition}[Uninterpreted Random Number Generator \textsc{(URNG)}]
\textsc{URNG} is an uninterpreted function

\begin{equation*}
f_{\textsc{rng}}(H, \ell_{\mathit{closed}}, u_{\mathit{closed}}, l, u) :  T_{1} \times \ldots \times T_{|D|} \times \mathbb{B} \times \mathbb{B} \times T \times T \to T,
\end{equation*}

\noindent where $T : \{\mathbb{Z}, \mathbb{R} \}$, $H$ is a collection of right side expressions extracted from the set of disequalities $D$, $\ell$ and $u$ determine the bounded interval for the randomly generated value, and $\ell_{\mathit{closed}}, u_{\mathit{closed}}$ are boolean flags that, when set, identify the corresponding bound as being closed. Without loss of generality, we require the following postconditions to hold, on any supplied implementation of $f_{\textsc{rng}}$:

\begin{enumerate}
\item $\forall h \in H. f_{\textsc{rng}}(H, \_, \_, \_, \_) \neq h$
\item $f_{\textsc{rng}}(H, \false,\false,\ell,u) \in (\ell,u)$
\item $f_{\textsc{rng}}(H, \false,\true,\ell,u) \in (\ell,u]$
\item $f_{\textsc{rng}}(H, \true,\false,\ell,u) \in [\ell,u)$
\item $f_{\textsc{rng}}(H, \true,\true,\ell,u) \in [\ell,u]$
\end{enumerate}
\label{def:uninterpRNG}
\end{definition}

The use of URNGs allows us to reason about valid regions of values for variable assignments instead of a particular value. Furthermore, the postconditions defined for these functions play an integral role in determining the soundness of the resulting Skolem function. It is important to note that we do not have to reason regarding the emptiness of the intervals. The intuition behind this is that such computed constraints infer an unrealizable contract. In these scenarios \aeval would declare the input $\forall\exists$-formula as invalid, and the Skolem extraction algorithm would never be invoked.

\subsection{Algorithm}
\label{sec:alg}
Alg.~\ref{alg:extractsk} shows our proposed procedure for extracting Skolem functions that allow for random choices.
It is invoked from Alg.~\ref{alg:ae_val} and takes a set of universally quantified variables $\vx$, an existentially quantified variable $y$, and Skolem constraints $\pi$ computed in Alg.~\ref{alg:ae_val}.
Alg.~\ref{alg:extractsk} constructs a graph of a function that is embedded in a relation, specified by a conjunction of expressions over the relational operators $\{=, \neq, >, \ge, \le, <\}$, using the following constraints:

\begin{align} \notag
E \eqdef \{y = f_i(x)\}_i  \quad
D \eqdef \{y \neq f_i(x)\}_i \quad
G \eqdef \{y > f_i(x)\}_i \cr
GE \eqdef \{y \ge f_i(x)\}_i  \quad
LE \eqdef \{y \le f_i(x)\}_i \quad
L \eqdef \{y < f_i(x)\}_i 
\end{align}

In addition to the constraints above, Alg.~\ref{alg:extractsk} also utilizes the following helper functions (where $\sim\,\in\{<,\le,=,\neq,\ge,>\}$):

\begin{gather*}
\textsc{ASN} (y \sim e(x)) \eqdef e \quad 
\textsc{MIN}(\{s\}) \eqdef \textsc{ASN} (s) \quad 
\textsc{MAX}(\{s\}) \eqdef \textsc{ASN} (s) \\
\textsc{MIN}(S) \eqdef \textit{ite}(\textsc{ASN} (s) \!\leqp\! \textsc{MIN}(S\!\setminus\!\{s\}), \textsc{ASN} (s) , \textsc{MIN}(S\!\setminus\!\{s\})), s \in S \\
\textsc{MAX}(S) \eqdef \textit{ite}(\textsc{ASN} (s)\!\geqp\! \textsc{MAX}(S\!\setminus\!\{s\}), \textsc{ASN} (s) , \textsc{MAX}(S\!\setminus\!\{s\})), s \in S \\
\end{gather*}

Operator \textsc{MIN} (\textsc{MAX}) computes a symbolic minimum (maximum) of the given set of constraints. While the algorithm is applicable for both LIA and LRA, the following transformations are used  for the case of integers:%

\begin{align}\notag
& \infer{A \le B-1}{A < B} \qquad
 \infer{A > B-1}{A \ge B} 
\end{align}

These transformations help avoid clauses containing $<$ and $\ge$.
Line~\ref{alg3:init} initializes the value of the boolean flags $\ell_{\mathit{closed}}$ and $u_{\mathit{closed}}$ to false, and line~\ref{alg3:trivcase} handles the case where equality constraints exist over $y$. Lines~\ref{alg3:boundchecks} to~\ref{alg3:boundchkend} construct the expressions for the lower and upper bounds, and the truth of the flags depends on the (symbolic) comparison between the symbolic minima and maxima. Line~\ref{alg3:lowequp} handles the case where the lower bound is equal to the upper bound. It should be noted that for cases handled by lines~\ref{alg3:trivcase} and~\ref{alg3:lowequp} only deterministic choices exists.

Lines~\ref{alg3:extractRHS} to~\ref{alg3:finalret} attempt to compute an expression containing a URNG that considers the set of disequalities $D$. First, the algorithm extracts the right-hand side of disequalities in line~\ref{alg3:extractRHS}. If both bounds are undefined, line~\ref{alg3:diseqbothundef} returns the application of the URNG $f_{\textsc{rng}}(H, \true, \true, -\infty, +\infty)$, where $-\infty$ and $+\infty$ are represented as free variables that can be later mapped respectively to the  minimum and maximum arithmetic representations supported by the implementation (e.g. \texttt{INT\_MIN} and \texttt{INT\_MAX} for integers in C). If only the lower bound is undefined (line~\ref{alg3:defu}), we use  $f_{\textsc{rng}}(H, \true, u_{\mathit{closed}}, -\infty, u)$ to generate a random value with an unconstrained lower bound. Similarly, we handle the case where no constraints exist for the upper bound in line~\ref{alg3:deflt}. In line~\ref{alg3:finalret}, both $\ell$ and $u$ are  defined and the algorithm returns  $f_{\textsc{rng}}(H, \ell_{\mathit{closed}}, u_{\mathit{closed}}, \ell, u)$ to capture a random value within the respective bounds. In all above cases, when $H \neq \varnothing$, the URNG is expected to generate a value that satisfies all disequality constraints in $D$. For the special case where $D = \varnothing$, there are no such disequalities over $y$ and the Skolem term can freely assign any value within the computed bounds $\ell$ and $u$.

As an illustration of our procedure, we present summarized runs over 
the following examples. 

\begin{example}
Consider the formula $\forall x. \exists y_1, y_2. \psi(x, y_1, y_2)$ over LIA, where:%

\begin{multline*}
\psi(x, y_1, y_2) \eqdef \\ (x \leq 2 \land y_1 > -3x \land y_2 < x) \lor(x \geq -1 \land y_1 < 5x \land y_2 > x)
\end{multline*}

The formula is valid since there exists an assignment to  $y_1$ and $y_2$ that satisfies the constraints in $\psi$, for any  $x$. In order to construct such a witness, \aeval needs to consider two separate cases for $x$, i.e., the constraints $x \leq 2$ and $x \geq -1$.

Under $x \leq 2$, the deterministic Skolem  terms would be $-3x + 1$  for $y_1$ and $x -1$ for $y_2$. 
For the random case, Alg.~\ref{alg:extractsk} computes $f_{\textsc{rng},y_1}(\varnothing, \false, \true, -3x, +\infty)$ and $f_{\textsc{rng},y_2}(\varnothing, \true, \false, -\infty, x)$. Under $x \geq -1$, the deterministic terms would be $-3x +1$ for $y_1$ and $x + 1$  for $y_2$, while Alg.~\ref{alg:extractsk} computes the functions $f_{\textsc{rng},y_1}(\varnothing, \true, \false, -\infty, 5x)$ and $f_{\textsc{rng}, y_2}(\varnothing, \false, \true, x, +\infty)$, respectively.

Note that for the random case, the above terms are finally combined into the Skolem term for $\exists y_1, y_2. \psi(x, y_1, y_2)$:

\begin{multline*}
\mathit{sk}_{\vec y}(x) \eqdef \textit{ite}(x \leq2, \\ (y_1 = f_{\textsc{rng},y_1}(\varnothing, \false, \true, -3x, +\infty) \land y_2 = f_{\textsc{rng},y_2}(\varnothing, \true, \false, -\infty, x)), \notag \\ (y_1 = f_{\textsc{rng},y_1}(\varnothing, \true, \false, -\infty, 5x) \land y_2 = f_{\textsc{rng}, y_2}(\varnothing, \false, \true, x, +\infty)))\notag
\end{multline*}
\label{ex:LIA}
\end{example}

\begin{example}
Consider an Assume-Guarantee contract for a system with the input vector $\vec x = \{x_1, x_2, x_3\} \in \mathbb{R}^3$ and one output $y \in \mathbb{R}$ and the following constraints

\begin{itemize}
\setlength\itemsep{.2em}
  \item $A(x_1, x_2) \eqdef x_1, x_2 \in (0, 1)$
  \item $G_I(y) \eqdef \true$
  \item $G_{T}(y, x_1, x_2, y') \eqdef y' \in (0,1) \land y' \neq x_1 \land y' \neq x_2$
\end{itemize}

The above specification is realizable as there are infinitely many assignments to $y$ that satisfy the guarantees $G$ given any value of $x_1, x_2$ in $(0, 1)$. Using Alg.~\ref{alg:extractsk} we retrieve the following Skolem term to enable random behavior (note that input $x_3$ is not included in the set $H$, i.e. the first argument of the function):

\begin{align}
\mathit{sk}_y(\vec x) \eqdef f_{\textsc{rng},y}(\{x_1,x_2\}, \false, \false, 0, 1) \notag.
\end{align}
\end{example}

\begin{example}

Consider the contract from Fig.~\ref{fig:contract}. The details of the synthesis procedure remain identical with the deterministic approach up until the Skolemization step. Fig.~\ref{fig:skolem} shows the C implementation for the random witness that is synthesized using Alg.~\ref{alg:extractsk}. Our proposed Skolemization procedure returns the assignment value for $y$ that is equivalent to $f_{\textsc{rng}}(\true, \true, -1, 1)$, for the conditions under which the system can safely choose to move the robot either left, right, or not at all. The actual choice is randomly made through the application of a function named \textit{RandVal}. The implementation of the function is then left to the engineer's discretion (an example is given in Fig.~\ref{fig:rng}). 

It is important to note that throughout this section, we presented Alg.~\ref{alg:extractsk} assuming that the input formula $\psi(\vec x, y)$ to \aeval is \emph{disjunction-free}.
The main difference between the original Skolemization procedure in \aeval and our new one is that the former does not provide all possible Skolem terms, while the latter does.
Thus, the original \aeval algorithm is not sensitive to the shape of the original formula, but our new algorithm requires a special treatment for the disjunctions.
In fact, our approach is generalizable to the case of arbitrary formulas via converting them the Disjunctive Normal Form (DNF).
The set of Skolem functions can then be generated for each of the disjunct separately.
The final solution is then composed from these \emph{partial} Skolem functions along with a parameter that randomly picks one of them.
\end{example}

\begin{figure}[!t]
\centering
\begin{lstlisting}[basicstyle=\small]
void skolem() {
  if(position + x == 1) {
    y = RandVal(1, 1, -1, 1);
  } else if (position + x >= -1 &&
  	position + x <= 0) {
    y = - (position + x);
  } else {
    y = RandVal(1, 1, -1, 1);
  }
}

\end{lstlisting}
\caption{Synthesized random witness.}
\label{fig:skolem}

\centering
\begin{lstlisting}[basicstyle=\small]
double RandVal(_Bool lflag, _Bool uflag,
  double lbound, double ubound){
  int min = lflag ? lbound : lbound+1;
  int max = uflag ? ubound : ubound-1;
  int range = max - min + 1;
  double rnd = (double) rand() /
  	(1.0 + (double) RAND_MAX);
  int value = (int) ((double) range * rnd);
  return value + min;
}
\end{lstlisting}
\caption{Example random number generator.}
\label{fig:rng}
\end{figure}

\subsection{Soundness and Completeness}

In this section we prove that Alg.~\ref{alg:extractsk} is sound and can provide all possible Skolem terms given a set of Skolem constraints $\pi(\vec x, y)$. As we noted in the beginning of Sect.~\ref{sec:alg}, Skolem constraints are created from literals of a formula in linear arithmetic, thus it cannot contain disjunctions.

\begin{theorem}[Soundness of Skolem Extraction]
Assuming that the properties 1-5 from Def.~\ref{def:uninterpRNG} hold, Alg.~\ref{alg:extractsk} returns valid Skolem terms. 
\label{thm:sound}
\end{theorem}

\begin{proof}
To prove this statement, it suffices to show that any computed Skolem term $\mathit{sk}_y(\vx)$ by Alg.~\ref{alg:extractsk} accompanied by the associated postconditions in Def.~\ref{def:uninterpRNG}, implies the input Skolem constraints in $\pi(\vx,y)$. Return lines~\ref{alg3:trivcase} and~\ref{alg3:lowequp} in \textsc{ExtractSK} are trivial cases, as they reduce to a simple assignment from equality constraints. Line~\ref{alg3:diseqbothundef} refers to the case where no bounds have been defined and the computed Skolem term is a URNG that utilizes the unconstrained variables $-\infty$ and $+\infty$ along with postcondition 5 to ensure the choice of an arbitrary value within the specified domain. Lines~\ref{alg3:defu} to~\ref{alg3:deflt} handle the case where inequalities exist that determine the lower and upper bounds $\ell(\vx)$ and $u(\vx)$. If the lower bound is undefined, line~\ref{alg3:defu} returns a URNG that is guaranteed to provide a random value between $-\infty$ and $u$ as per postconditions 4 and 5. We prove the soundness of terms provided by line~\ref{alg3:deflt} in a similar manner. If both bounds exist, then in line~\ref{alg3:finalret} the Skolem term returned is a URNG guaranteed to provide a value within the range specified by $\ell(\vx), u({\vx})$, as per postconditions 2-5. 
\end{proof}

\begin{theorem}[Completeness of Skolem Extraction]
The Skolem terms generated by Alg.~\ref{alg:extractsk} are sufficient to represent all possible witnesses of the conjunctive $\forall\exists$-formula in Eq.~\ref{eq:viable}.
\label{thm:complete}
\end{theorem}

\begin{proof}
It suffices to prove that no weaker set of postconditions $pc'$ (i.e., $pc \Rightarrow pc'$) exists, such that:

\begin{equation}
\forall \vx. pc'(\mathit{sk}(\vx)) \Rightarrow \pi(\vx, \mathit{sk}(\vx))
\label{eq:weakprec}
\end{equation}

\noindent
We prove this by contradiction, assuming that  $pc'$ exists whenever Alg.~\ref{alg:extractsk} returns.

\noindent\textbf{Lines~\ref{alg3:trivcase} and ~\ref{alg3:lowequp}}. Alg.~\ref{alg:extractsk} returns the deterministic assignments $\textsc{ASN}(e)$ and $\ell$, for which no weaker postconditions exist.

\noindent\textbf{Line~\ref{alg3:diseqbothundef}}. In this case, no bounds have been defined, and postcondition 5 is used to denote a range with unconstrained bounds $-\infty$ and $+\infty$. Formally, we can simplify this postcondition to $pc = true$, for which no weaker postcondition exists. It is also noteworthy to state that weaker postconditions exist, but all of them have to violate at least one disequality in $D$.

\noindent\textbf{Line~\ref{alg3:defu}}. We have $\ell =$ undef, i.e., no constraints exist for the lower bound, and the Skolem term 

\begin{equation*}
\mathit{sk}(\vx) = f_{\textsc{rng}}(H, \true, u_{\mathit{closed}}, \mathit{min}, u)
\end{equation*}

is returned. Depending on whether the upper bound $u$ is closed or not, we have two cases. For brevity, we show the proof for the case where $u$ is closed, and the corresponding case for the open bound follows similar principles.
\begin{itemize}
    \item When $u$ is closed, the output constraints are simplified to $\pi(\vx, \mathit{sk}(\vx)) = \mathit{sk}(\vx) \leq u$ and the Skolem term
\begin{equation*}
\mathit{sk}(\vx) = f_{\textsc{rng}}(H, \true, \true, -\infty, u)
\end{equation*}

is returned, with postcondition 5 capturing the term's range. Assume that a weaker postcondition $pc'$ exists, such that Eq.~\ref{eq:weakprec} holds. Without loss of generality, we pick 

\begin{equation*}
pc' = f_{\textsc{rng}}(H,\true,\true, -\infty, u) \in [-\infty, u']
\end{equation*}

with $u' > u$. Therefore, we have that $pc \Rightarrow pc'$, but Eq.~\ref{eq:weakprec} does not hold for $pc'$, as the new term may provide the value $u'$ as an output,  falsifying $\pi(\vx, \mathit{sk}(\vx))$.
    
\end{itemize}

\noindent\textbf{Line~\ref{alg3:deflt}}. Similar to proof for line~\ref{alg3:defu}.

\noindent\textbf{Line~\ref{alg3:finalret}}. $\ell \neq u \neq \text{undef}$, and as such the output constraints can be simplified into $\pi(\vx, \mathit{sk}(\vx)) = \ell \sim \mathit{sk}(\vx) \sim u$, where $\sim{} \in \{<, \leq\}$. We have the following cases corresponding to the possible ranges:

\begin{enumerate}
  \item $(\ell,u)$. In this case we have $\mathit{sk}(\vx) = f_{\textsc{rng}}(H, \false, \false, \ell, u)$ and as postcondition $pc$ the second postcondition from Def~\ref{def:uninterpRNG}. Assume that a weaker postcondition $pc'$ exists, such that Eq.~\ref{eq:weakprec} holds. We can pick $pc' = f_{\textsc{rng}}(H, \false, \false, \ell, u) \in [\ell,u]$. In this case, $pc \Rightarrow pc'$ holds, but Eq.~\ref{eq:weakprec} does not hold, as we can pick any of the assignments $\mathit{sk}(\vx) = \ell$, $\mathit{sk}(\vx) = u$, which violate the constraints in $\pi$, reaching a contradiction.
  \item $[\ell, u)$. Similar to the previous proof, by picking, e.g., $pc' = f_{\textsc{rng}} \in [\ell,u]$.
  \item $(\ell, u]$. Similarly, we can pick $pc' = f_{\textsc{rng}} \in [\ell, u]$.
  \item $[\ell, u]$. Similarly, we can pick $pc' = f_{\textsc{rng}} \in [\ell', u]$, where $\ell' < \ell$.
\end{enumerate}
\end{proof}
\section{Implementation and Evaluation}
\label{sec:results}

We implemented our random synthesis algorithm as a complementary procedure to the original synthesis framework \jkind~\cite{gacek2018jk}, a Java implementation of a popular \textsc{Kind} model checker~\cite{katis2018validity,gacek2018jk,champion2016kind}. Following \textsc{Kind}, the input contracts are expressed using the Lustre dataflow language~\cite{lustrev6}.
\jkind provides support for synthesis both through the fixpoint algorithm in \jsynvg as well as its predecessor, \jsyn, a realizability checking algorithm based on the $k$-induction principle~\cite{gacek2015towards,katis2016towards,katis2015machine}. Our proposed Skolemization procedure in Alg.~\ref{alg:extractsk} is a new extraction method that is performed after the validity checking procedure in \aeval, thus making it inherently compatible with both \jsyn and \jsynvg~\footnote{The modified version of \jsynvg for random synthesis is available at \url{https://github.com/andrewkatis/jkind-1/tree/synthesis}. The modified version of \aeval with the new Skolemization procedure is available at \url{https://github.com/andrewkatis/fuzzersynthesis}.}. It is noteworthy that our approach does not add any performance overhead to the baseline implementation of \jsynvg, as shown in Table~\ref{tbl:comp}.

Since the synthesized Skolem functions are expressed in the SMT-LIB 2.0 language~\cite{barrett2010smt} by default, we translate them into executable C implementations. For the purposes of this paper, we mapped the application of URNGs to calls to random number generators of uniformly distributed values, unless otherwise noted.

The evaluation process of our work is twofold:

\begin{enumerate} 
 \item \emph{Empirical}. We performed case studies in applications where synthesis of random designs can be beneficial.\footnote{\href{https://figshare.com/articles/software/Synthesis_of_Infinite-State_Systems_with_Random_Behavior/12228026}{The benchmarks are publicly accessible online.} DOI : 10.6084/m9.figshare.12228026} For the first case study, we conducted an experiment in the context of model-based fuzz testing, where the goal was to synthesize \emph{reactive graybox fuzzers} capable of exposing vulnerabilities that can crash an application, through random test case generation. The second study revolves around \emph{controller synthesis} for avoidance games in robot motion planning.
  \item \emph{Synthesis time}. We investigated the effect that our Skolemization algorithm had on \jsynvg in terms of synthesis time. Furthermore, we compared our work to \textsc{DT-Synth}, a state-of-the-art synthesis tool for infinite-state problems~\cite{neider2019learning}.
\end{enumerate}

\section{Case Study 1:  Reactive Fuzzers}
\label{sec:fuz}

In our first case study, we explored the applicability of synthesized implementations with random behavior in fuzz testing. We focused on model-based approaches to examine a system-under-test (SUT), the input specification of which was used to derive test cases (see Utting et al. for a detailed survey~\cite{utting2012taxonomy}). In the past, model-based fuzz testing revolved around the use of structured descriptions of the system input in the form of grammars and a sophisticated implementation of a fuzzer that, given a grammar, would continuously feed random inputs to the SUT~\cite{aschermann2019nautilus, peachfuzzer, pham2016model, veggalam2016ifuzzer}. We show that synthesis offers a viable alternative technique in this context, where the generated implementations can serve as SUT-specific fuzzers, requiring for configuration nothing but the input specification for the SUT.

\subsection{Setup and Evaluation}

\begin{figure}[!t]
\centering
\includegraphics[width=3in]{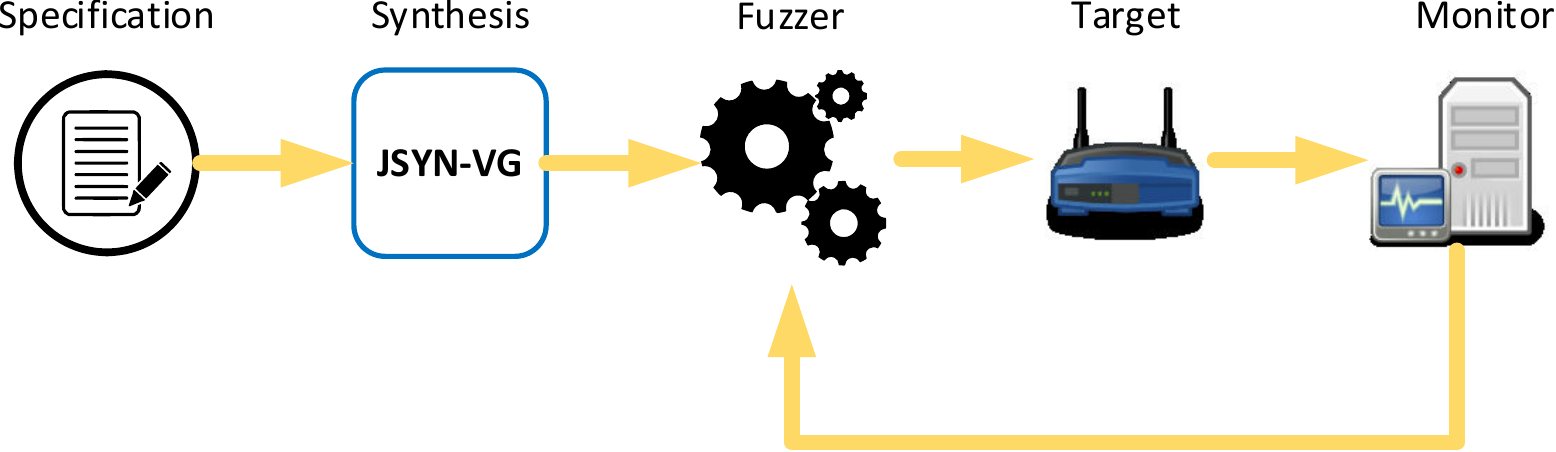}
\caption{Fuzzer synthesis and testing diagram}
\label{fig:fuzzsetup}
\end{figure}

The main intuition is that the SUT's input description can be viewed as a substantial fragment of the fuzzer's specification, which can then be used to synthesize a reactive random test case generator. Fig.~\ref{fig:fuzzsetup} depicts our exact setup, where the designer already has a specification for the SUT and uses \jsynvg with our Skolemization algorithm to automatically generate a corresponding fuzzer. The fuzzer is then attached to the SUT (\textsf{Target}), along with an accompanied monitoring service (\textsf{Monitor}) that tracks progress with respect to the SUT-related statistics (e.g., coverage). Following the definition of \emph{graybox fuzzing}, a feedback loop exists where monitored information can be subsequently fed to the fuzzer, in order to dictate the generation of future test cases. 

\begin{table*}[t!]
\centering
\caption{Fuzzer performance comparison and synthesis times.}
\label{tbl:fuzzcomp}
\resizebox{\textwidth}{!}{
\begin{tabular}{|c|c|c|c|c|c|c|c|c|c|c|c|}
\hline
\multirow{2}{*}{} & \multicolumn{4}{c|}{\textsc{AFL}} & \multicolumn{4}{c|}{\textsc{AFLFast}} & \multicolumn{3}{c|}{\multirow{2}{*}{Synthesized Fuzzer}} \\ \cline{2-9}
 & \multicolumn{2}{c|}{w/o corpus} & \multicolumn{2}{c|}{w/ corpus} & \multicolumn{2}{c|}{w/o corpus} & \multicolumn{2}{c|}{w/ corpus} & \multicolumn{3}{c|}{} \\ \hline
\textbf{System Under Test} & \textbf{Coverage (\%)} & \textbf{Crashed?} & \textbf{Coverage (\%)} & \textbf{Crashed?} & \textbf{Coverage (\%)} & \textbf{Crashed?} & \textbf{Coverage (\%)} & \textbf{Crashed?} & \textbf{Coverage (\%)} & \textbf{Crashed?} & \textbf{Synthesis times (s)} \\ \hline
basic\_messaging & \textbf{83.76\%} & \xmark & 82.48\% & \xmark & 82.48\% & \xmark & 82.48\% & \xmark & 81.21\% & \cmark & 16.067 \\ \hline
Dive\_Logger & 92.18\% & \xmark & 92.41\% & \xmark & 92.18\% & \xmark & 92.18\% & \xmark & \textbf{93.79\%} & \cmark & 99.852 \\ \hline
Divelogger2 & 80.50\% & \xmark & \textbf{87.04\%} & \xmark & 78.40\% & \xmark & 85.08\% & \xmark & 83.25\% & \xmark & 77.689 \\ \hline
Email\_System\_2 & 78.71\% & \xmark & \textbf{92.90\%} & \xmark & 91.61\% & \xmark & \textbf{92.90\%} & \xmark & 84.84\% & \cmark & 36.624 \\ \hline
Movie\_Rental\_Service & 38.32\% & \xmark & 38.72\% & \xmark & 38.32\% & \xmark & 38.72\% & \xmark & \textbf{49.50\%} & \cmark & 140.826 \\ \hline
Palindrome & \textbf{78.13\%} & \cmark & \textbf{78.13\%} & \cmark & \textbf{78.13\%} & \cmark & \textbf{78.13\%} & \cmark & 75.00\% & \cmark & 1.231 \\ \hline
PTaaS & 71.94\% & \cmark & \textbf{78.06\%} & \cmark & 46.76\% & \cmark & 46.76\% & \cmark & 74.46\% & \cmark & 77.041 \\ \hline
Quadtree\_Conways & 67.04\% & \xmark & \textbf{84.08\%} & \xmark & 67.04\% & \xmark & 67.04\% & \xmark & 64.79\% & \xmark & 88.743 \\ \hline
SCUBA\_Dive\_Logging & 80.97\% & \cmark & 80.67\% & \cmark & 76.41\% & \cmark & 76.41\% & \cmark & \textbf{83.56\%} & \cmark & 101.813 \\ \hline
User\_Manager & 58.43\% & \xmark & 67.45\% & \xmark & 29.02\% & \xmark & 29.02\% & \xmark & \textbf{79.80\%} & \xmark & 16.289 \\ \hline
\end{tabular}
}
\end{table*}

Using this setup, we proceeded with a thorough performance evaluation of our synthesized fuzzers, following guidelines that were recently proposed by Klees et al.~\cite{klees2018evaluating}:

\noindent\textbf{SUT Selection}. We considered ten applications from the DARPA Cyber Grand Challenge (CGC)~\footnote{The public CGC benchmark collection is available at \url{https://bit.ly/2HBqrJq}.}, a benchmark collection that has been extensively used in the past to assess the performance of fuzzers due to the high degree of interactivity between the SUT and the user~\cite{stephens2016driller, rawat2017vuzzer, peng2018t}. The original collection was aimed towards the evaluation of automated reasoning and testing tools, and each application is intentionally documented in a way that is insufficient to derive a precise specification from the documents themselves. To simulate the context under which synthesis would make most sense as a tool, we closely inspected and ran each application in order to identify the types and sequences of inputs each application takes. The manual process of discovering and writing input specifications per application was non-trivial, as each application differs considerably from the rest, and was the main factor to the study being limited to a subset of ten applications.
  
\noindent\textbf{Fuzzer specification}. After inspecting each application to identify the kinds of inputs that it takes, we wrote a corresponding Assume-Guarantee contract for a fuzzer. Each fuzzer specification consists of properties that capture the valid ranges of values for each one of the SUT inputs. Moreover, the specification is stateful, making each fuzzer reactive to changes (or lack thereof) in coverage results from previously generated tests. We specified the behavior of the fuzzer in such a way that, for the majority of its runtime, valid inputs are fed into the SUT. As long as no progress is made in terms of coverage, the fuzzer attempts generating invalid tests with probability $p = 0.2$. 
  
\noindent\textbf{Formalization}. All of the aforementioned elements that comprise the fuzzer specification can be expressed using a set of \emph{safety} properties over the SUT inputs, 
  where each set precisely captures the conditions under which a (in)valid value is generated for the corresponding input. An example pair of such properties is the following:

\begin{itemize}
\item $prop_1 \eqdef p' \ge 0 \land p' \le 1$
\item $prop_2 \eqdef (\lnot cvg \land (p \le 0.1 \lor p \ge 0.9)) \Rightarrow in'_{\text{sys}} \notin S_{\text{valid}}$
\end{itemize}

Variables $p$ and $in_{\text{sys}}$ are fuzzer outputs, with $in_{\text{sys}}$ also serving as a corresponding input for the SUT. The value of $p \in [0,1]$ is picked randomly for each test, and it determines whether the next (primed) system input $in'_{\text{sys}}$ will be assigned to a valid value (i.e., a value in $S_{\text{valid}}$) or not. Variable $cvg$ is an input to the fuzzer and can be viewed as a flag which, when set, informs the fuzzer that the previous test resulted in progress in system coverage (e.g., line coverage improved). If such progress was not observed, then we allow the fuzzer to randomly consider invalid values in subsequent tests. More specifically, when $p \le 0.1 \lor p \ge 0.9$, the fuzzer will generate an invalid value, i.e., a value that does not satisfy the constraints that define $S_{\text{valid}}$. Following the notation that we described in previous sections, the synthesis problem for the properties above is to ensure that $\forall p, cvg, in_{\text{sys}} \exists p', in'_{\text{sys}}. (prop_1 \land prop_2)$ is valid. 

\noindent\textbf{Synthesis and Evaluation}. Using the fuzzer contracts, we synthesized a fuzzer for each application and ran it against the SUT using the setup in Fig.~\ref{fig:fuzzsetup}. We set the timeout for each fuzzing campaign to nine hours, and monitored the SUT line coverage (\emph{gcov}) as well as crashes. To compare performance, we also ran fuzzing campaigns using \textsc{AFL}~\cite{zalewskiAFL} and \textsc{AFLFast}~\cite{bohme2017coverage}, using their default configurations. We selected these tools primarily due to \textsc{AFL} being one of the most prominent tools in the area, while \textsc{AFLFast} is a recent extension to \textsc{AFL} that has been shown to perform better with respect to vulnerability detection.\footnote{Both \textsc{AFL} and \textsc{AFLFast} do not support line coverage reporting natively. To monitor coverage, we used \emph{afl-cov}~\cite{aflcov}, a wrapper tool that enables the use of \emph{gcov} with \textsc{AFL} and its variants.} Both tools were run both with and without an initial corpus in order to provide a more complete picture of their performance, whether the user provides additional information or not. To remain fair with respect to the evaluation, the corpora were created using tests that exercise application locations that are as deep as possible.


Table~\ref{tbl:fuzzcomp} shows the results of our experiments. Most of the applications contain unreachable code related to debugging methods, and as such 100\% coverage is not attainable using \emph{gcov} without further modifications to the source code. While we were able to achieve $\geq 75\%$ line coverage for the majority of the benchmarks, the application ``Movie\_Rental\_Service'' was the worst performing with only 49.5\%. Despite that, the synthesized fuzzer outperformed both \textsc{AFL} and \textsc{AFLFast} on either configuration with a significant margin. In fact, our synthesized fuzzers outperformed both \textsc{AFL} and \textsc{AFLFast} on four applications and remained within 4\% of the best performing tool for five others, with ``Quadtree\_Conways'' being the only exception. More interestingly, seven of the synthesized fuzzers were able to crash the corresponding application at least once, whereas \textsc{AFL}/\textsc{AFLFast} were only able to crash three.

Considering the performance results along with the low synthesis time per fuzzer, we believe that synthesis of model-based fuzzers should be considered a viable tactic towards testing systems where a specification already exists. Arguably, a synthesized fuzzer is as easy to use as a general-purpose tool like \textsc{AFL}. Furthermore, the user does not have to provide additional information through a corpus, a procedure in testing that often times can be time consuming and cumbersome, as both valid and invalid input sequences have to be considered for a successful campaign. 

\section{Case Study 2: Robot Motion Planning}
\label{sec:robot}

In our second study, we synthesized implementations for robots participating in two-player safety games against an adversary. The study is furthermore split into two parts.

\subsection{Simulating avoidance games}
\label{subsec:avoid}
We experimented on simulating an avoidance game in a bounded arena, where the synthesized solution was used against two different adversarial scenarios. Both the properties of the robot and the adversary were specified using their position in terms of $(x,y)$ coordinates. Formally, we described the game using the following properties:

\begin{itemize}[leftmargin=*]
  \item Initial state : The robot starts in $(x_{\text{init}}, y_{\text{init}})$ (similarly for the adversary).
  \item Valid transitions : $x'_{\text{robot}} \in [x_{\text{robot}} - \delta, x_{\text{robot}} + \delta]$, where $\delta$ is user-defined and captures the maximum distance between subsequent moves (similarly for $y$-coordinate and the adversarial transitions).
  \item In-bounds property : $x_{\text{robot}} \ge x_{\text{min}} \land x_{\text{robot} }\le x_{\text{max}}$ (similarly for the $y$-coordinate).
  \item Avoidance property : $x_{\text{robot}} \neq x_{\text{adversary}} \lor y_{\text{robot}} \neq y_{\text{adversary}}$.
\end{itemize}

The first scenario in our presentation involves the adversary patrolling on a specific route, while in the second the adversary is always moving towards the robot. Trajectory videos for both scenarios are available online\footnote{Pictures and videos of the simulated games presented in this section were anonymized and made available at \url{https://figshare.com/s/ce2dfd885b3caf20f46d}.}.

\subsubsection{Real Coordinates.}
Fig.~\ref{fig:traj} shows three possible trajectories that were generated after running the synthesized solution for 1000 turns against the patrolling adversary. Both robots move in the arena using rational coordinates in a 5x5 box. The initial location for the robot is the point $(0.5,0.5)$ and the adversary begins its route from $(0.8,0.8)$. While the adversary has a predetermined route, the robot is allowed to move towards any possible direction (vertically, horizontally and diagonally). Moreover, the robot can move at varying distances up to 0.1 units away from its current position, in both axes (i.e., $|x_{\text{robot}} - x'_{\text{robot}}| \leq 0.1$, and similarly for the y axis). Fig.~\ref{fig_first_case} indicates how the synthesized solution can respond in a random pattern, covering different parts of the bounded arena while preserving safety. Fig.~\ref{fig_second_case} and~\ref{fig_third_case} demonstrate the resulting trajectories when the user introduces bias in the values returned by the random number generators, using the same generated witness from \aeval. As a result, the robot was limited to moves that would retain its position within the central area of the arena (Fig.~\ref{fig_second_case}) and close to the bottom left corner of the patrolling adversary's route (Fig.~\ref{fig_third_case}).


\begin{figure*}[!t]
\centering
\subfloat[]{\includegraphics[width=2.3in]{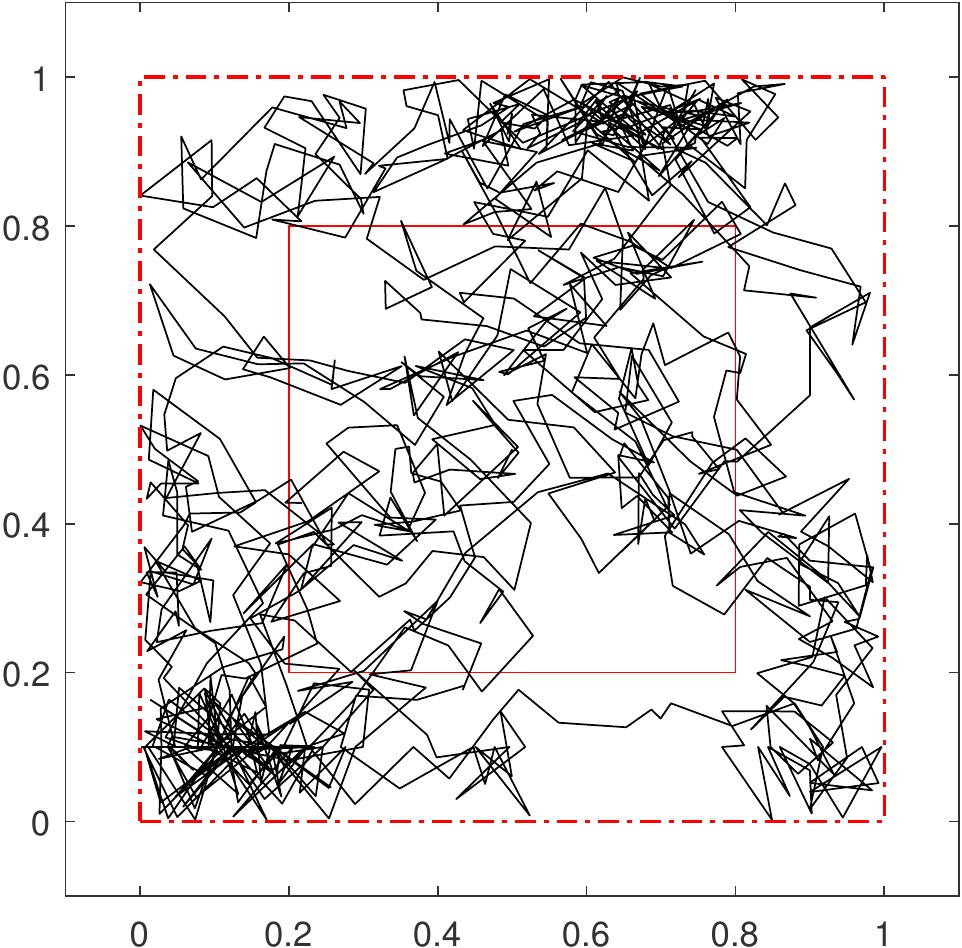}
\label{fig_first_case}}
\hfil
\subfloat[]{\includegraphics[width=2.3in]{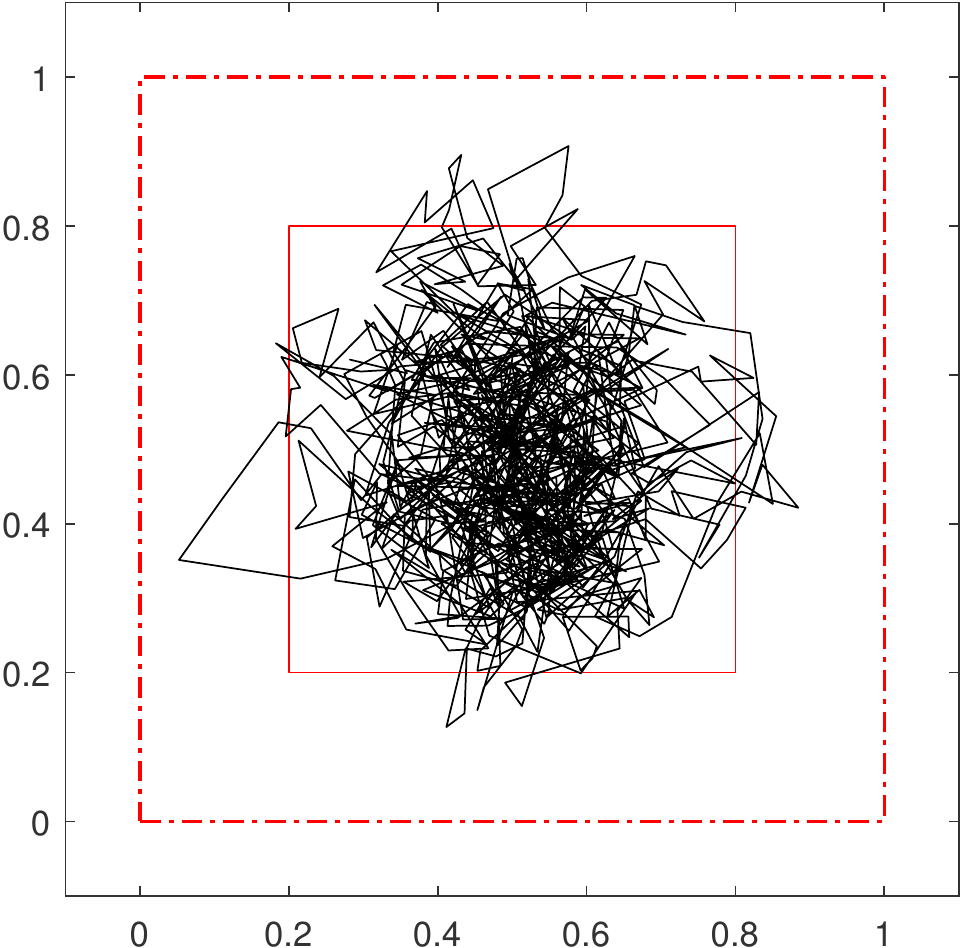}
\label{fig_second_case}}
\hfil
\subfloat[]{\includegraphics[width=2.3in]{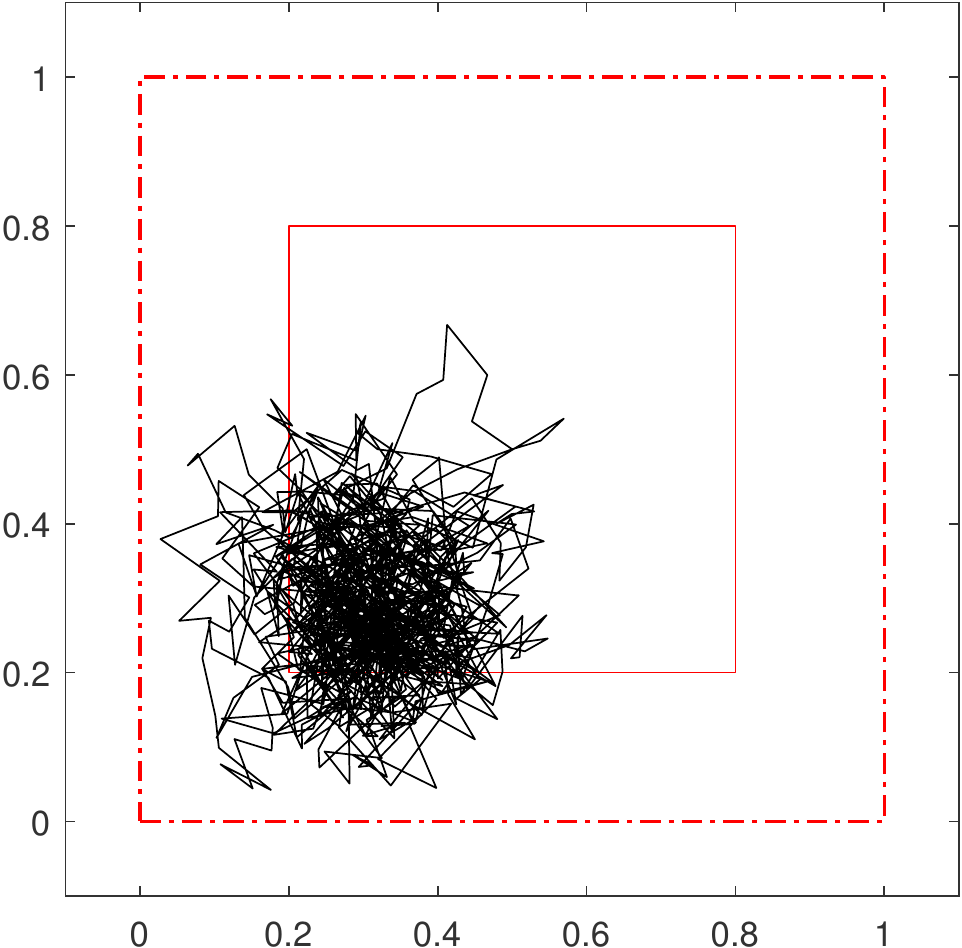}
\label{fig_third_case}}
\caption{Random trajectories of a robot (irregular solid line) while avoiding a patrolling adversary (inner square).}
\label{fig:traj}
\end{figure*}

\subsubsection{Integer Coordinates.}
For this experiment, we aimed to demonstrate the advantages that randomness can provide with respect to how well a robot covers a bounded arena, inspired by work in coverage path planning problems. Fig.~\ref{fig:nondetvsdet} shows how two trajectories evolved over several turns (100, 250 and 1000 turns) for a similar motion planning problem using integer coordinates. To demonstrate which parts of the arena the robot explored we outline its trajectory with a bold black line, while the red line represents the trajectory of the adversary. In this game, the adversary is aggressively chasing after the robot in a random fashion. The robot's objective remains the same, i.e., move within the bounded arena while avoiding the adversary. The robot's initial location is the point $(0,0)$, while the adversary begins at $(6,6)$. 

In fact, Fig.~\ref{fig_first_nondet_case},~\ref{fig_second_nondet_case}, and~\ref{fig_third_nondet_case} show moves performed by a random controller, while Fig.~\ref{fig_first_det_case},~\ref{fig_second_det_case},and ~\ref{fig_third_det_case} depict the behavior of the deterministic solution provided by the standard synthesis algorithm in \jsynvg. It is apparent that 
the former visits 100\% states in less than 250 turns, whereas the latter visits only 30\% states in 1000 turns. 
This comparison showcases the advantages that a random solution can provide in terms of overall coverage as well as the diversity of behaviors that can be observed and exercised when an implementation can be generated that always considers the entire set of safe choices, instead of an instantiated strategy.


\begin{figure}[!ht]
\centering
\subfloat[100 turns]{\includegraphics[width=1.5in]{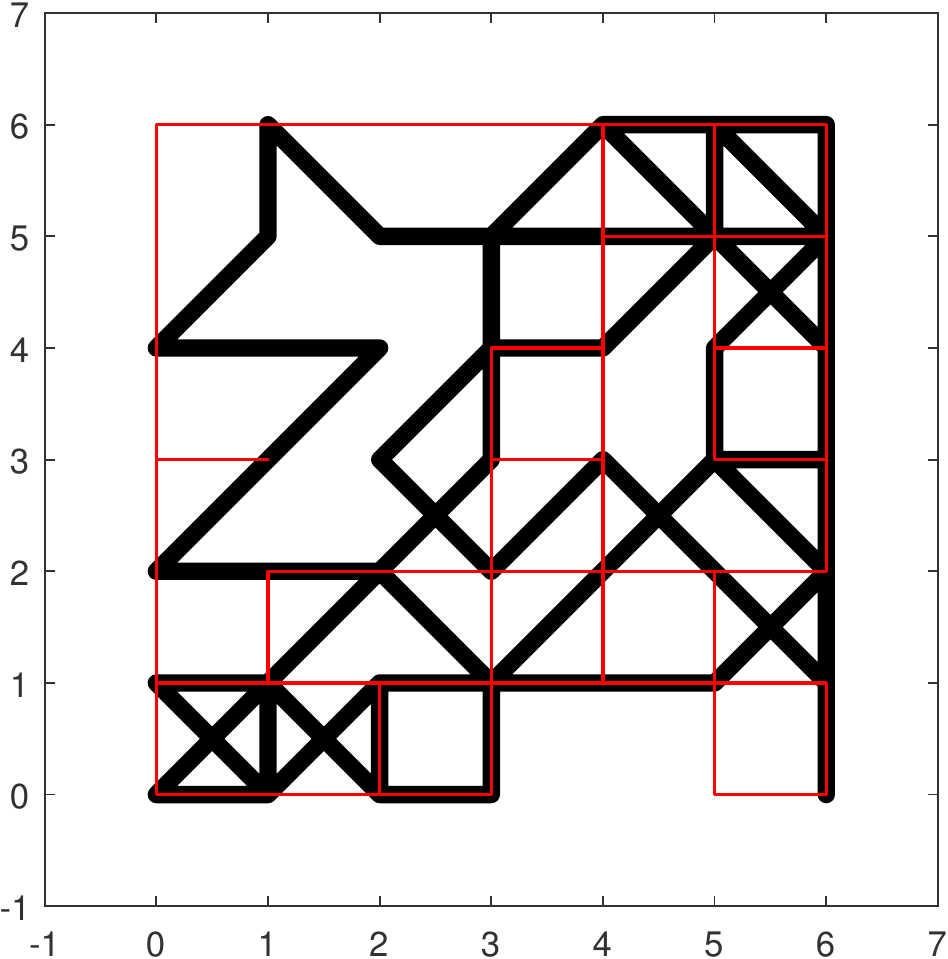}
\label{fig_first_nondet_case}}
 \hfil
\subfloat[100 turns]{\includegraphics[width=1.5in]{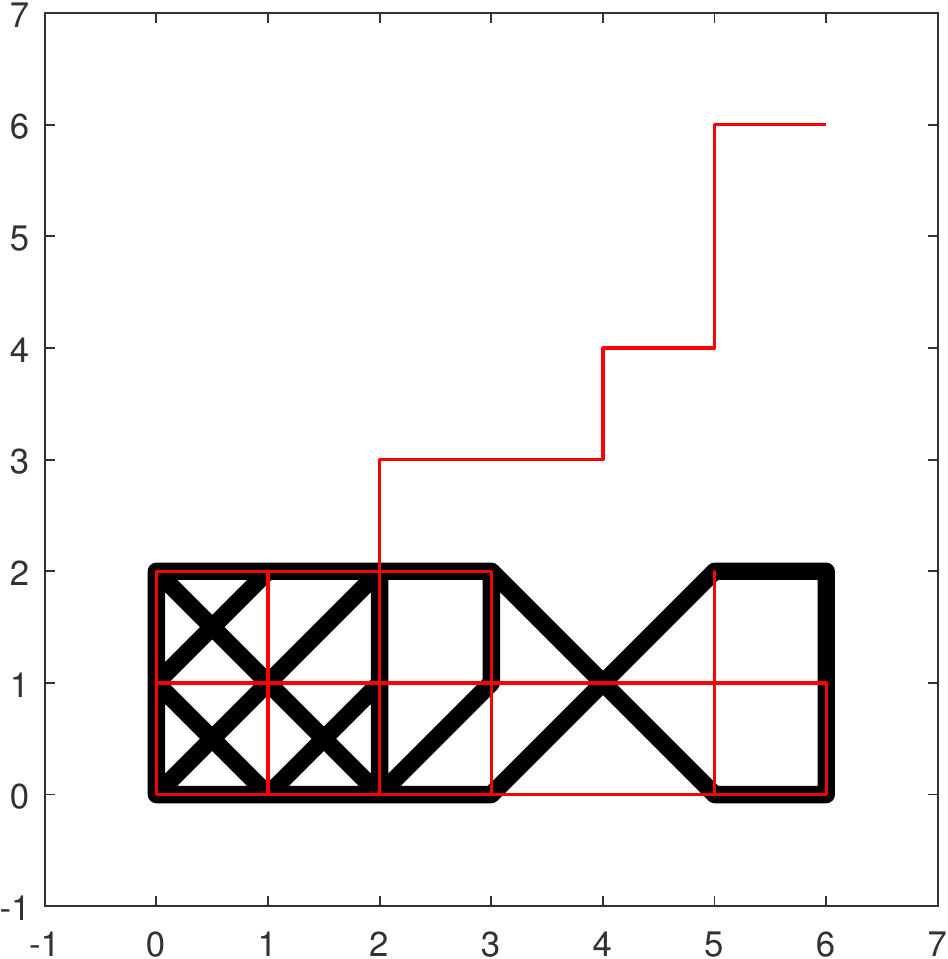}
\label{fig_first_det_case}}
\\
\subfloat[250 turns]{\includegraphics[width=1.5in]{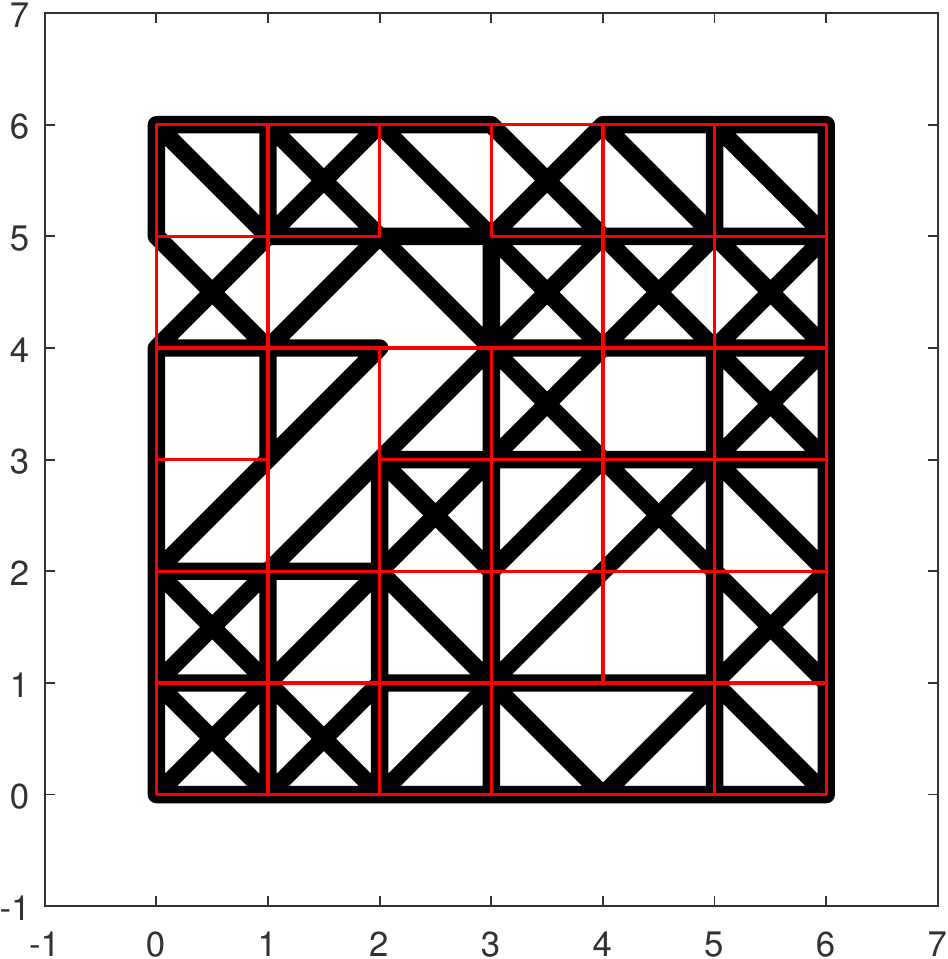}
\label{fig_second_nondet_case}}
 \hfil
\subfloat[250 turns]{\includegraphics[width=1.5in]{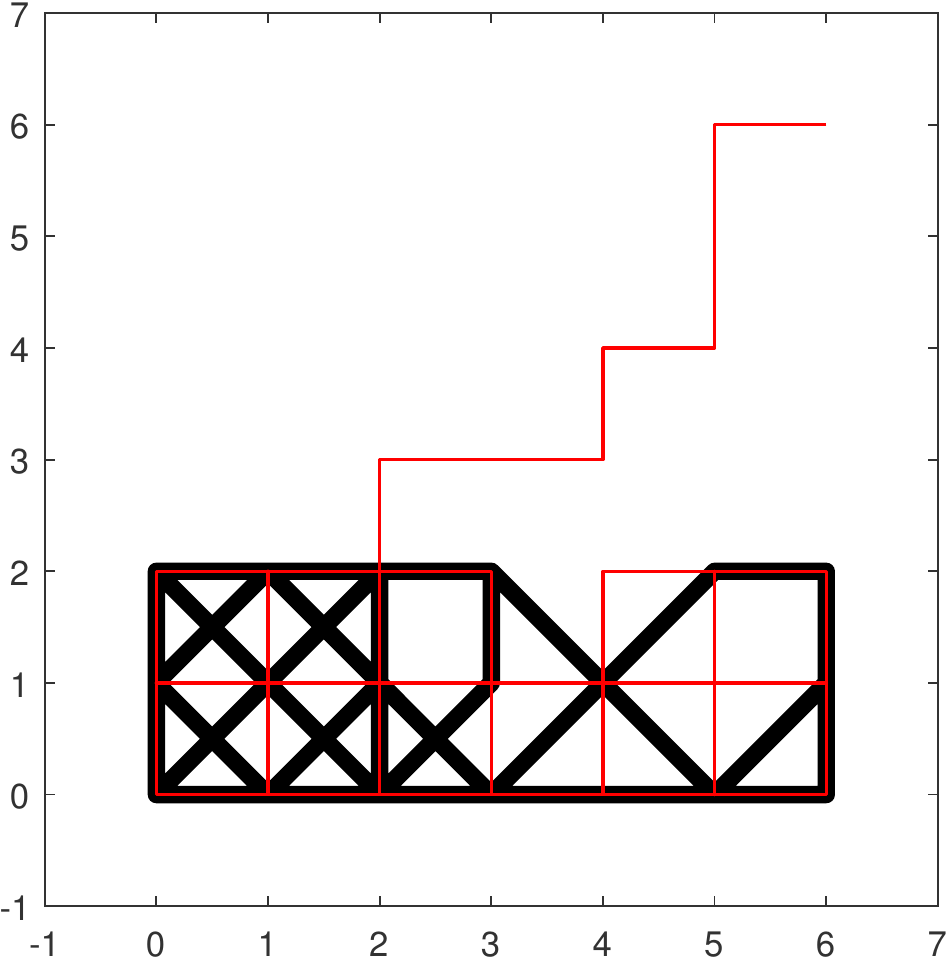}
\label{fig_second_det_case}}
\\
\subfloat[1000 turns]{\includegraphics[width=1.5in]{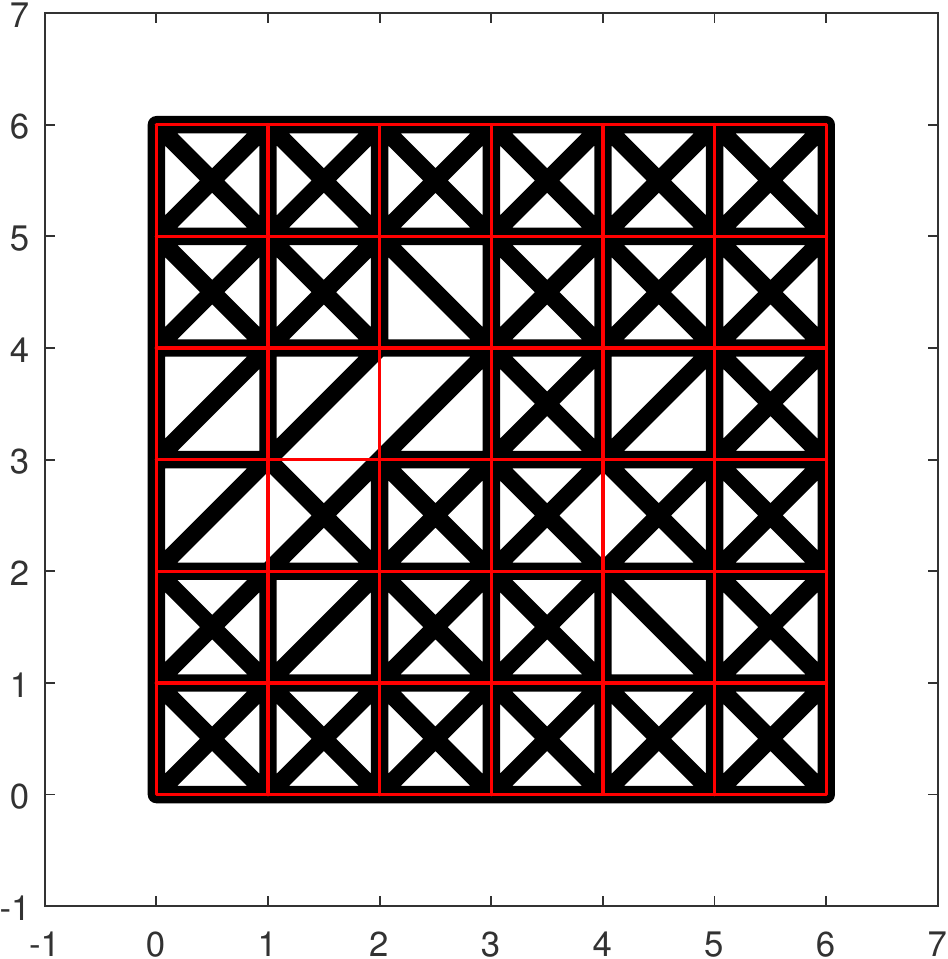}
\label{fig_third_nondet_case}}
 \hfil
\subfloat[1000 turns]{\includegraphics[width=1.5in]{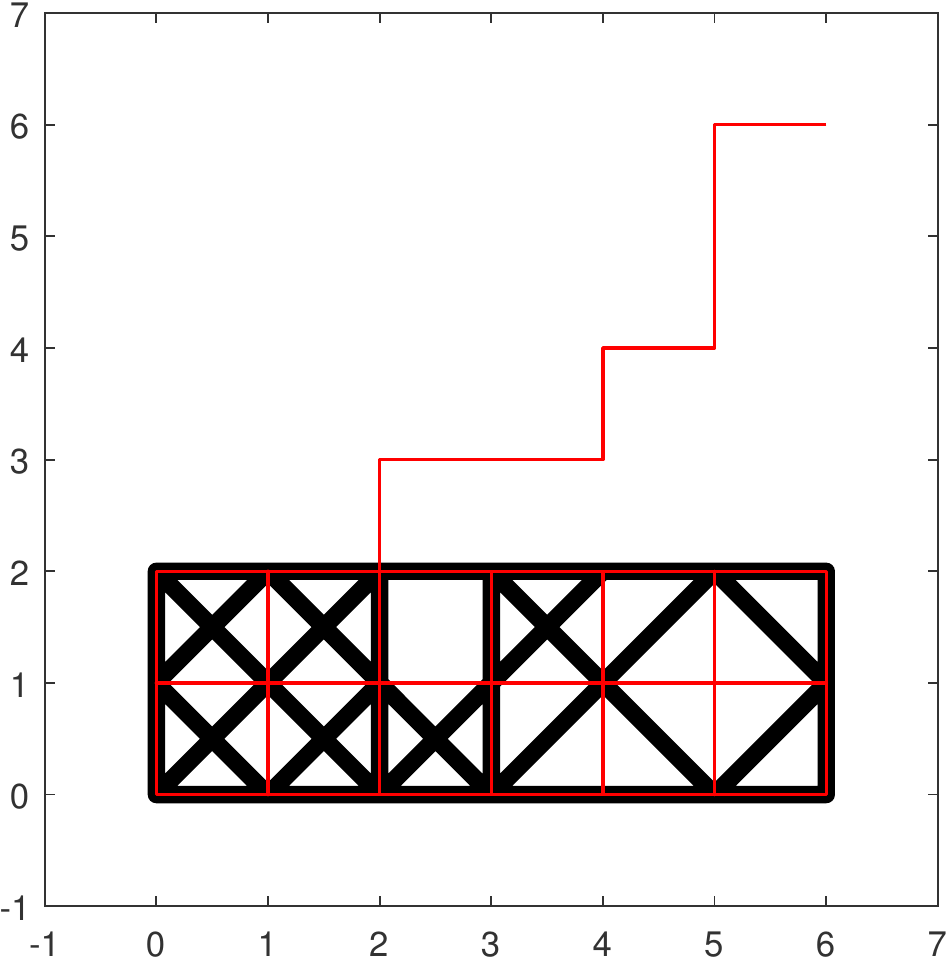}
\label{fig_third_det_case}}
\caption{Trajectories and area coverage over time of random (left column) versus deterministic (right column) controller.}
\label{fig:nondetvsdet}
\end{figure}

\section{Evaluation -- Synthesis time}

Our case study in robot motion planning was inspired by results in this context from the most recent and related work on \textsc{DT-Synth}~\cite{neider2019learning}. This reactive synthesis framework incorporates learning techniques to generate winning sets for infinite-state safety games in the form of decision trees. \textsc{DT-Synth} has been shown to outperform previous proposed synthesis tools, both in infinite-state (\textsc{ConSynth}~\cite{beyene2014constraint} and finite-state problems (\textsc{RPNI-Synth} and \textsc{SAT-Synth}~\cite{neider2016automaton}). While the authors do not explicitly talk about randomness, the winning sets provided by \textsc{DT-Synth} are sufficient to generate implementations with diverse behavior. Despite this fact, the generated winning sets are subsets of the greatest fixpoint of safe states, which would lead to implementations that only exercise a fragment of the reachable state space. An example is the winning set that we mentioned for the motivating example in Section~\ref{sec:ex}.

Note that \textsc{DT-Synth} works only for finite-branching game graphs, and the user must additionally specify a minimum value for the number of successors for each vertex in the graph. An incorrect value for this threshold can lead to unsound witnesses. With our \jsynvg, such additional knowledge is not required from the user since it is only reliant on the original specification and is guaranteed to provide sound results, thanks to Theorem~\ref{thm:sound}.

Table~\ref{tbl:comp} presents the comparison of \jsynvg and \textsc{DT-Synth}. As an addendum we included the synthesis times for the problems using the existing deterministic synthesis algorithm in \jsynvg. As we mentioned in Sect.~\ref{sec:results}, the performance is identical when compared to synthesizing random witnesses.

For the purposes of this comparison, we used the infinite-state benchmarks presented in the original paper on \textsc{DT-Synth}~\cite{neider2019learning}, as well as the simulated avoidance games from Sect.~\ref{subsec:avoid}, namely \emph{bounded\_evasion} and \emph{bounded\_evasion\_ints}. The two tools have similar performance for half the benchmarks, with significant differences for the rest. For \emph{diagonal}, the main distinction is that \textsc{DT-Synth} requires the definition of two additional expressions to guide the learning procedure, whereas \jsynvg finds a solution without additional templates. On the other hand, \textsc{DT-Synth}'s ability to synthesize memoryless strategies allows for  faster synthesis for \emph{solitarybox}, where the robot is simply moving freely within an infinite arena while staying within a horizontal stripe of width equal to three. \jsynvg is targeted to synthesis of stateful systems, and as such, a more elaborate strategy is generated. 

In the case of~\emph{repair-critical}, \textsc{DT-Synth} appears to be more efficient in terms of handling disjunctive expressions in the specification, while for~\emph{synth-synchronization} \textsc{DT-Synth} seems to require more elaborate hypotheses in order to come up with a witness. The latter is further demonstrated in the results for the~\emph{cinderella} and~\emph{bounded\_evasion\_ints} games, where \textsc{DT-Synth} fails to synthesize a witness within the timeout of 15 minutes. In contrast, \jsynvg  computes a greatest fixpoint of safe states and synthesizes a solution in a few seconds. Finally, for the game~\emph{bounded\_evasion}, \textsc{DT-Synth} does not currently support the theory of linear real arithmetic. 

\begin{table}[t!]
\centering
\caption{Synthesis time of \textsc{DT-Synth} and \jsynvg (seconds).}
\label{tbl:comp}
\begin{tabular}{|l|c|c|c|}
\hline
\textbf{\textbf{Benchmark}} & \textbf{\jsynvg} & \textbf{\begin{tabular}[c]{@{}c@{}}\jsynvg \\ (random)\end{tabular}} & \textbf{DT-Synth} \\ \hline
box & 0.603 & 0.606 & \textbf{0.258} \\ \hline
diagonal & 1.109 & \textbf{1.011} & 6.027 \\ \hline
evasion & 0.705 & 0.605 & 0.660 \\ \hline
follow & 3.34 & 3.029 & \textbf{1.034} \\ \hline
limitedbox & 3.229 & 3.332 & 3.350 \\ \hline
solitarybox & 1.902 & 1.816 & \textbf{0.284} \\ \hline
square & 5.823 & \textbf{5.605} & 6.44 \\ \hline
program-repair & 3.122 & 3.638 & \textbf{2.452} \\ \hline
repair-critical & 83.891 & 88.073 & \textbf{30.593} \\ \hline
synth-synchronization & 23.013 & \textbf{23.2} & 89.804 \\ \hline
cinderella ($c = 2$) & 20.061 & \textbf{20.167} & $>900$ \\ \hline
cinderella ($c = 3$) & 12.02 & \textbf{11.294} & $>900$ \\ \hline
bounded\_evasion & 49.528 & \textbf{49.662} & unsupported \\ \hline
bounded\_evasion\_ints & 31.614 & \textbf{32.611} & $>900$ \\ \hline
\end{tabular}
\end{table}
\section{Related Work}
\label{sec:related}

The idea of synthesizing reactive designs with random behavior is relevant to synthesis for permissive games. This area has been explored in the past for finite-state problems~\cite{bouyer2009measuring, konighofer2017shield}. More recently, Fremont and Seshia described a formal extension to the theory of Control Improvisation to support reactive synthesis~\cite{fremont2018reactivecontrolimprov}. Their probabilistic approach is limited to finite-state problems and practically useful only for the subset of safety games. The end result is a maximally-randomized finite word generator, called an improviser, where each word satisfies the predetermined probability threshold constraints. In comparison, our approach synthesizes designs for infinite-state problems that simulate randomness. Furthermore, we do not provide guarantees regarding the randomness of the responses from the synthesized witness. Instead, we focus on synthesizing witnesses that consider regions of values as candidates to variable assignments, a problem reducible to SMT. In our case, the end product is an implementation that can be further refined by the engineer with a custom probability distribution to retrieve random values. This provides invaluable freedom as the user can then choose whether to express bias through the requirements or through the random number generators themselves.

The original work on \jsynvg targeted the area of infinite-state problems. In this context, Beyenne et al. first proposed a template-based approach called \textsc{ConSynth}, where the specification is accompanied by a template regarding the shape of the solution to guide the synthesizer towards a solution~\cite{DBLP:conf/popl/BeyeneCPR14}. In contrast, \jsynvg is a completely automated approach exempting the user from necessity to further reason about the shape of the computed solution and allowing to focus on expressing the problem in the form of input-output contracts. Permissive solutions for infinite-state games have primarily been proposed in the context supervisor synthesis~\cite{claessen2017supervisory,shoaei2014supervisory}, where a controller is synthesized considering a formal representation of the behavior (inputs) provided by the participating plant. Compared to this work, our proposed solution explores the applicability of synthesized controllers with random behavior, while the overall synthesis task is inherently harder due to not requiring an exact implementation of the controller's environment.

Neider and Markgraf recently proposed \textsc{DT-Synth}~\cite{neider2019learning}, a learning-based approach to synthesizing winning sets for infinite state games in the form of decision trees, as an extension to previous work by Neider and Topcu for finite-state problems~\cite{neider2016automaton}. \textsc{DT-Synth} requires additional knowledge regarding the number of successor states, where lack thereof can lead to unsound results. Similarly to \textsc{ConSynth}, for more complex problems, the game specification is supported by additional syntactic expressions that help the learner converge faster to a solution. In contrast, our Skolem extraction algorithm is guaranteed to provide sound witnesses and does not depend on additional user-provided input.
\section{Conclusion}
\label{sec:conclusion}

We have presented a novel Skolemization procedure for the \aeval solver enabling the synthesis of infinite-state reactive implementations with random behavior. The proposed solution is an extension to the synthesis algorithm \jsynvg that computes a greatest fixpoint of safe states. The product is a witness where values inside safe regions are being considered as equally safe candidate assignments. Our solution provides the engineer with flexibility when it comes to introducing additional bias through the specification or the implemented random number generators.

To the best of our knowledge, this is the first work that is capable of synthesizing random infinite-state systems. We showed how the extended synthesis framework can be effectively used to synthesize promising solutions in the context of robot motion planning, as well as a novel application in fuzz testing. In the future, we wish to continue exploring the area of reactive fuzzer synthesis, particularly in the context of identifying a formal specification standard. To expand its applicability in robot motion planning, we wish to explore ways to support liveness specifications, as well as soft requirements. The outstanding result of this work is a Skolem extraction procedure general enough to be applicable to other, unexplored, aspects of the synthesis problem.

\begin{acks}
The authors would like to thank Kristopher Cory for his invaluable advice towards setting up the fuzz testing experiments. The authors would also like to thank the anonymous reviewers for their comments. This work was funded in part by the \grantsponsor{}{United States Department of the Navy, Office of Naval Research}{} under contract \grantnum{}{N68335-17-C-0238}.
\end{acks}

\bibliographystyle{ACM-Reference-Format}
\bibliography{main}
\end{document}